\documentclass[11pt]{article}
\usepackage{bcc24}

\usepackage{amstext, mathtools}
\usepackage{tikz}
\usepackage{pgfplots}
\usepackage{hyperref}
\usepackage[titlenumbered, linesnumbered, ruled]{algorithm2e}

\makeatletter
\def\input@path{{./Figures/}}
\makeatother

\usetikzlibrary{patterns,arrows,automata,positioning,decorations,shapes,calc}

\tikzstyle{normalvertex}=[circle,fill=white,draw=black]
\tikzstyle{emptyvertex}=[draw,circle,minimum size=7pt,inner sep=0pt]
\tikzstyle{tinyvertex}=[draw,circle,minimum size=3pt,inner sep=0pt]
\tikzstyle{thickedge}=[draw,gray!60,line width=1.6pt,-]

\DeclareMathOperator{\id}{id}

\DeclareMathOperator{\Iso}{Iso}
\DeclareMathOperator{\Aut}{Aut}
\DeclareMathOperator{\Sym}{Sym}
\DeclareMathOperator{\Alt}{Alt}

\DeclareMathOperator{\Aff}{Aff}

\DeclareMathOperator{\dist}{dist}

\DeclareMathOperator{\mgamma}{\widehat{\Gamma}}

\DeclareMathOperator{\Br}{Br}
\DeclareMathOperator{\Unf}{Unf}

\DeclareMathOperator{\cl}{cl}

\newcommand{\CO}{\mathcal O}
\newcommand{\CE}{\mathcal E}

\newcommand{\CH}{\mathcal H}

\newcommand{\FA}{\mathfrak A}
\newcommand{\FB}{\mathfrak B}
\newcommand{\FP}{\mathfrak P}

\newcommand{\Fx}{\mathfrak x}
\newcommand{\FX}{\mathfrak X}
\newcommand{\Fy}{\mathfrak y}
\newcommand{\FY}{\mathfrak Y}
\newcommand{\Fz}{\mathfrak z}

\newcommand{\NN}{\mathbb N}

\DeclareMathOperator{\width}{wd}
\DeclareMathOperator{\tw}{tw}

\newcommand{\bigmid}{\mathrel{\big|}}
\newcommand{\Bigmid}{\mathrel{\Big|}}

\newcommand{\Biglmulti}{{\Big\{\hspace{-5pt}\Big\{}}
\newcommand{\Bigrmulti}{{\Big\}\hspace{-5pt}\Big\}}}

\newcommand{\polylog}[1]{\operatorname{polylog}(#1)}

\newcommand{\ColRef}[1]{\chi_{(\infty)}(#1)} % #1 - graph
\newcommand{\ColRefIt}[2]{\chi_{(#1)}(#2)} % #1 - iteration, #2 - graph
\newcommand{\ColRefItVe}[3]{\chi_{(#1)}(#2)(#3)} % #1 - iteration, #2 - graph, #3 - vertex
\newcommand{\WL}[2]{\chi^{#1}_{(\infty)}(#2)}
\newcommand{\WLVe}[3]{\chi^{#1}_{(\infty)}(#2)(#3)}
\newcommand{\WLIt}[3]{\chi^{#1}_{(#2)}(#3)}
\newcommand{\WLItVe}[4]{\chi^{#1}_{(#2)}(#3)(#4)}

\newcommand{\tColRef}[2]{\chi_{#1\text{-}{\sf CR}}(#2)} % #1 - graph

\newcommand{\atp}{\operatorname{atp}}
\newcommand{\CC}{{\mathcal C}}
\newcommand{\LC}{\textsf{\upshape C}}

\bcctitle{Recent Advances on the Graph Isomorphism Problem}
\bccshorttitle{Recent Advances on the Graph Isomorphism Problem}

\bccname{Martin Grohe and Daniel Neuen}
\bccshortname{M.\ Grohe and D.\ Neuen} 

\bccaddressa{RWTH Aachen University\\Ahornstr.\ 55\\52074 Aachen\\Germany}
\bccemaila{grohe@informatik.rwth-aachen.de}

\bccaddressb{CISPA Helmholtz Center for Information Security\\Stuhlsatzenhaus 5\\66123 Saarbruecken\\Germany}
\bccemailb{daniel.neuen@cispa.saarland}
\begin{document}
\makebcctitle

\begin{abstract}
 We give an overview of recent advances on the graph isomorphism problem.
 Our main focus will be on Babai's quasi-polynomial time
 isomorphism test and subsequent developments that led to the design
 of isomorphism algorithms with a quasi-polynomial parameterized running
 time of the form $n^{\polylog k}$, where $k$ is a graph parameter such as the maximum degree.
 A second focus will be the combinatorial Weisfeiler-Leman algorithm. 
\end{abstract}

\section{Introduction}

Determining the computational complexity of the Graph Isomorphism
Problem (GI) is regarded a major open problem in computer
science. Already stated as an open problem in Karp's seminal paper
on the NP-completeness of combinatorial problems in 1972~\cite{kar72},
GI remains one of the few natural problems in NP that is neither known to be
NP-complete nor known to be polynomial-time solvable. 
In a breakthrough five years ago, Babai~\cite{Babai16} proved
that GI is solvable in \defword{quasi-polynomial time}
$n^{\polylog n} \coloneqq n^{(\log n)^{\CO(1)}}$, where $n$ denotes the number of vertices of the
input graphs. This was the first improvement on the worst-case running time of a
general graph isomorphism algorithm over the $2^{\CO(\sqrt{n\log n})}$ bound
established by Luks in 1983 (see \cite{BabaiKL83}).

Research on GI as a computational problem can be traced back to work
on chemical information systems in the
1950s~(e.g.~\cite{raykir57}). The first papers in the computer science
literature, proposing various heuristic approaches, appeared in the
1960s (e.g.~\cite{ung64}). An early breakthrough was Hopcroft and
Tarjan's $\CO(n\log n)$ isomorphism test for planar graphs
\cite{hoptar72}. Isomorphism algorithms for specific graph classes have
been an active research strand ever since
(e.g.~\cite{filmay80,gromar15,lin92,Luks82,mil80,pon88}).

Significant progress on GI was made in the early 1980s when Babai and Luks
introduced non-trivial group-theoretic methods to the field
\cite{bab79,BabaiL83,Luks82}. In particular, Luks \cite{Luks82} 
introduced a general divide-and-conquer framework based on the
structure of the permutation groups involved. It is the basis for
most subsequent developments. Luks used it to design an isomorphism test
running in time $n^{\CO(d)}$, where $d$ is the maximum degree of the
input graphs. The divide-and-conquer framework is also the foundation of Babai's
quasi-polynomial time algorithm. Primitive permutation groups
constitute the difficult bottleneck of Luks's strategy. Even though we
know a lot about the structure of primitive permutation groups (going
back to \cite{Cameron81} and the classification of finite simple
groups), for a long a time it was not clear how we can exploit this
structure algorithmically. It required several groundbreaking ideas by
Babai, both on the combinatorial and the group-theoretic side, to
obtain a quasi-polynomial time algorithm.  Starting from Luks's
divide-and-conquer framework, we will describe the group-theoretic
approach and (some of) Babai's new techniques in Section~\ref{sec:group}.

In subsequent work, jointly with our close collaborators Pascal
Schweitzer and Daniel Wiebking, we have designed a number of
\defword{quasi-polynomial parameterized algorithms} where we confine the
poly-logarithmic factor in the exponent to some parameter that can be
much smaller than the number $n$ of vertices of the input graph. The
first result \cite{GroheNS18} in this direction was an isomorphism
test running in time $n^{\polylog d}$ which improves both Luks's
and Babai's results (recall that $d$ denotes the maximum degree of the
input graph). After similar quasi-polynomial algorithms parameterized
by tree width \cite{Wiebking20} and genus \cite{Neuen20}, the work
culminated in an isomorphism test running in time $n^{\polylog k}$ for
graphs excluding some $k$-vertex graph as a minor \cite{GroheNW20}. In
Section~\ref{sec:parameterized-isomorphism-tests}, we will present
these results in a unified framework that is based on a new
combinatorial closure operator and the notion of \defword{$t$-CR-bounded}
graphs, introduced by the second author in \cite{Neuen20} (some ideas can be
traced back as far as \cite{Miller83b,Ponomarenko86}).

The combinatorial algorithm underlying this approach is a very simple
partition refinement algorithm known as the Color Refinement
algorithm (first proposed in \cite{mor65}), which is also an important
subroutine of all practical graph isomorphism algorithms.
The Color Refinement algorithm can be viewed as the
1-dimensional version of a powerful combinatorial algorithm known as
the Weisfeiler-Leman algorithm (going back to \cite{weilem68}). The
strength of the Weisfeiler-Leman algorithm and its multiple, sometimes
surprising connections with other areas, have been a recently very
active research direction that we survey in Section~\ref{sec:wl}.

It is not the purpose of this paper to cover all aspects of GI.
In fact, there are many research strands that we
completely ignore, most notably the work on practical isomorphism
algorithms. Instead we want to give the reader at least an impression
of the recent technical developments that were initiated by Babai's
breakthrough. For a broader and much less technical survey, we
refer the reader to \cite{groschwe20}.

%%%%%%%%%%%%%%%%%%%%%%%%%%%%%%%%%%%%%%%%%%%%%%%%%%%%%%%%%%%%
\section{Preliminaries}\label{sec:prel}
%%%%%%%%%%%%%%%%%%%%%%%%%%%%%%%%%%%%%%%%%%%%%%%%%%%%%%%%%%%%
Let us review some basic notation. By $\NN$, we denote the natural
numbers (including $0$), and for every positive $k\in\NN$ we let
$[k] \coloneqq \{1,\ldots,k\}$. Graphs are always finite and, unless explicitly
stated otherwise, undirected. We denote the vertex set and edge set of
a graph $G$ by $V(G)$ and $E(G)$, respectively. For (undirected) edges, we write
$vw$ instead of $\{v,w\}$.
The \defword{order} of a graph $G$ is $|G|\coloneqq|V(G)|$.

A \defword{subgraph} of a graph $G$ is a graph $H$ with
$V(H)\subseteq V(G)$ and $E(H)\subseteq E(G)$. For a set
$W\subseteq V(G)$, the \defword{induced subgraph} of $G$ with vertex set
$W$ is the graph $G[W]$ with $V(G[W])=W$ and
$E(G[W])\coloneqq\{vw\in E(G)\mid v,w\in W\}$. Moreover,
$G - W\coloneqq G[V(G)\setminus W]$. The set of \defword{neighbors} of a vertex
$v$ is $N_G(v)\coloneqq\{w\mid vw\in E(G)\}$, and the \defword{degree} of $v$ is
$\deg_G(v)\coloneqq|N_G(v)|$. If $G$ is clear from the context, we omit the
subscript ${}_G$. The \defword{maximum degree} of $G$ is
$\Delta(G)\coloneqq\max\{\deg(v)\mid v\in V(G)\}$.

A \defword{homomorphism}
from a graph $G$ to a graph $H$ is a mapping $h:V(G)\to V(H)$ that
preserves adjacencies, that is, if $vw\in V(G)$ then $h(v)h(w)\in
V(H)$. An isomorphism from $G$ to $H$ is a bijective mapping $\varphi:V(G)\to V(H)$ that
preserves adjacencies and non-adjacencies, that is, $vw\in E(G)$ if and only if $\varphi(v)\varphi(w)\in E(H)$.
We write $\varphi\colon G \cong H$ to indicate that $\varphi$ is an isomorphism from $G$ to $H$.

Sometimes, it will be convenient to work with \defword{vertex- and
  arc-colored graphs}, or just \defword{colored graphs},
$G=(V(G),E(G),\chi^G_V,\chi^G_E)$, where $(V(G),E(G))$ is a graph,
$\chi_V^G\colon V(G) \rightarrow C_V$ is a vertex coloring, and
$\chi_E^G\colon \{(v,w) \mid vw \in E(G)\} \rightarrow C_E$ an arc
coloring. We speak of an \defword{arc coloring} rather than edge
coloring, because we color ordered pairs $(v,w)$ and not undirected
edges $vw$. Homomorphisms and isomorphisms of colored graphs must
respect colors.

%%%%%%%%%%%%%%%%%%%%%%%%%%%%%%%%%%%%%%%%%%%%%%%%%%%%%%%%%%%%
\section{The Weisfeiler-Leman Algorithm}\label{sec:wl}
%%%%%%%%%%%%%%%%%%%%%%%%%%%%%%%%%%%%%%%%%%%%%%%%%%%%%%%%%%%%

\subsection{The Color Refinement Algorithm}

One of the simplest combinatorial procedures to tackle GI is the \defword{Color Refinement algorithm}. 
The basic idea is to label vertices of the graph with their iterated degree sequence.
More precisely, an initially uniform coloring is repeatedly refined by
counting, for each color, the number of neighbors of that color.

Let $\chi_1,\chi_2 \colon V(G) \rightarrow C$ be colorings of vertices
of a graph $G$, where $C$ is some finite set of colors.  The coloring
$\chi_1$ \defword{refines} $\chi_2$, denoted $\chi_1 \preceq \chi_2$, if
$\chi_1(v) = \chi_1(w)$ implies $\chi_2(v) = \chi_2(w)$ for all
$v,w \in V(G)$.  The colorings $\chi_1$ and $\chi_2$ are
\defword{equivalent}, denoted $\chi_1 \equiv \chi_2$, if
$\chi_1 \preceq \chi_2$ and $\chi_2 \preceq \chi_1$.

The Color Refinement algorithm takes as input a vertex- and
arc-colored graph $(G,\chi_V,\chi_E)$ and outputs an
isomorphism-invariant coloring of the vertices.  Initially, the algorithm
sets $\ColRefIt0G\coloneqq \chi_V$.  Then the coloring is
iteratively refined as follows.  For $i > 0$ we let
$\ColRefItVe iGv \coloneqq \big(\ColRefItVe{i-1}Gv,\mathcal{M}_i(v)\big)$,
where
\[
  \mathcal{M}_i(v) \coloneqq\Biglmulti
  \big(\ColRefItVe{i-1}Gw,\chi_E(v,w),\chi_E(w,v)\big)
  \Bigmid w \in N_G(v)\Bigrmulti
\]
is the multiset of colors for the neighbors computed in the previous
round.

In each round, the algorithm computes a coloring that is finer than
the one computed in the previous round, i.e.,
$\ColRefIt iG \preceq \ColRefIt{i-1}G$.  At some point this
procedure stabilizes, meaning the coloring does not become strictly
finer anymore.  In other words, there is an $i_\infty < n$
such that $\ColRefIt{i_\infty}G\equiv\ColRefIt{i_\infty+1}G$.  We call
$\ColRefIt{i_\infty}G$ the \defword{stable coloring} and denote it by
$\ColRef G$.

\begin{example}
 The sequence of colorings for a path of length $6$ is shown in Figure \ref{fig:cr-path}.
\end{example}

\begin{figure}
  \centering
   \begin{tikzpicture}
  
  \node[anchor=east] at (-0.5,0) {$\ColRefIt0{P_6}$};
  \foreach \i in {0,...,6}{
   \node[emptyvertex,fill=white] (v\i) at (\i,0) {};
  }
  
  \node at (6.7,-0.5) {
   \scalebox{0.3}{
   \begin{tikzpicture}[rotate=110]
    \draw[white,fill=black!20] (-30:1cm) arc (-30:-180:1cm) -- (-180:1.25cm) -- (-.75,.75) -- (-180:.25cm) -- (-180:.5cm) arc (-180:-30:.5cm)--cycle;
   \end{tikzpicture}
   }
  };
  
  \node[anchor=east] at (-0.5,-1) {$\ColRefIt1{P_6}$};
  \node[emptyvertex,fill=black] (w0) at (0,-1) {};
  \node[emptyvertex,fill=white] (w1) at (1,-1) {};
  \node[emptyvertex,fill=white] (w2) at (2,-1) {};
  \node[emptyvertex,fill=white] (w3) at (3,-1) {};
  \node[emptyvertex,fill=white] (w4) at (4,-1) {};
  \node[emptyvertex,fill=white] (w5) at (5,-1) {};
  \node[emptyvertex,fill=black] (w6) at (6,-1) {};
  
  \node at (6.7,-1.5) {
   \scalebox{0.3}{
   \begin{tikzpicture}[rotate=110]
    \draw[white,fill=black!20] (-30:1cm) arc (-30:-180:1cm) -- (-180:1.25cm) -- (-.75,.75) -- (-180:.25cm) -- (-180:.5cm) arc (-180:-30:.5cm)--cycle;
   \end{tikzpicture}
   }
  };
  
  \node[anchor=east] at (-0.5,-2) {$\ColRefIt2{P_6}$};
  \node[emptyvertex,fill=black] (x0) at (0,-2) {};
  \node[emptyvertex,fill=gray] (x1) at (1,-2) {};
  \node[emptyvertex,fill=white] (x2) at (2,-2) {};
  \node[emptyvertex,fill=white] (x3) at (3,-2) {};
  \node[emptyvertex,fill=white] (x4) at (4,-2) {};
  \node[emptyvertex,fill=gray] (x5) at (5,-2) {};
  \node[emptyvertex,fill=black] (x6) at (6,-2) {};
  
  \node at (6.7,-2.5) {
   \scalebox{0.3}{
   \begin{tikzpicture}[rotate=110]
    \draw[white,fill=black!20] (-30:1cm) arc (-30:-180:1cm) -- (-180:1.25cm) -- (-.75,.75) -- (-180:.25cm) -- (-180:.5cm) arc (-180:-30:.5cm)--cycle;
   \end{tikzpicture}
   }
  };
  
  \node[anchor=east] at (-0.5,-3) {$\ColRefIt 3{P_6}=\ColRef {P_6}$};
  \node[emptyvertex,fill=black] (y0) at (0,-3) {};
  \node[emptyvertex,fill=gray] (y1) at (1,-3) {};
  \node[emptyvertex,fill=gray!40] (y2) at (2,-3) {};
  \node[emptyvertex,fill=white] (y3) at (3,-3) {};
  \node[emptyvertex,fill=gray!40] (y4) at (4,-3) {};
  \node[emptyvertex,fill=gray] (y5) at (5,-3) {};
  \node[emptyvertex,fill=black] (y6) at (6,-3) {};
  
  \foreach \u in {v,w,x,y}{
   \path[draw,thick] (\u0) -- (\u1);
   \path[draw,thick] (\u1) -- (\u2);
   \path[draw,thick] (\u2) -- (\u3);
   \path[draw,thick] (\u3) -- (\u4);
   \path[draw,thick] (\u4) -- (\u5);
   \path[draw,thick] (\u5) -- (\u6);
  }
 \end{tikzpicture}
 \caption{The iterations of Color Refinement on a path $P_6$ of length $6$.}
 \label{fig:cr-path}
\end{figure}
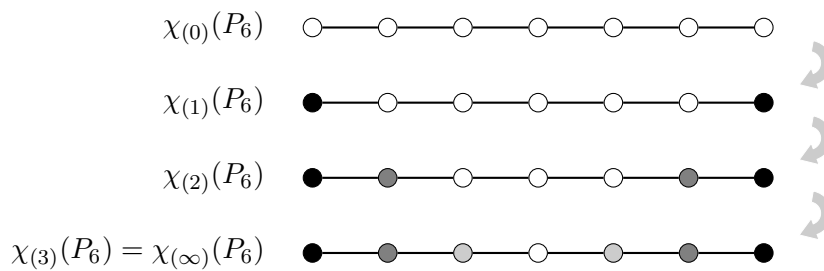

The Color Refinement algorithm is very efficient, and it is used as an important
subroutine in all practical isomorphism tools
\cite{CodenottiKSM13,JunttilaK11,Lopez-PresaCA14,McKay81,McKayP14}.
The stable coloring of a graph can be computed in time
$\CO((n+m)\log n)$, where $n$ denotes the number of vertices and $m$ the
number of edges of the input graph \cite{carcroC82} (also
see~\cite{paitar87})%
\footnote{To be precise, the algorithms compute the partition of the
  vertex set corresponding to the stable coloring, not the actual
  colors viewed as multisets.}.
For a natural class of partitioning algorithms,
this is best-possible~\cite{berbongro17}.

\subsection{Higher Dimensions}
The \defword{$k$-dimensional Weisfeiler-Leman algorithm ($k$-WL)} is a
natural generalization of the Color Refinement algorithm where, instead of
vertices, we color $k$-tuples of vertices of a graph. The
$2$-dimensional version, also known as the \defword{classical}
Weisfeiler-Leman algorithm, has been introduced by Weisfeiler and Leman
\cite{weilem68}, and the $k$-dimensional generalization goes
back to Babai and Mathon (see~\cite{CaiFI92}).

To introduce $k$-WL, we need the notion of the \defword{atomic type}
$\atp(G,\bar v)$ of a tuple $\bar v=(v_1,\ldots,v_\ell)$ of vertices
in a graph $G$. It describes the isomorphism type of the labeled
induced subgraph $G[\{v_1,\ldots,v_\ell\}]$, that is, for graphs $G,H$
and tuples $\bar v=(v_1,\ldots,v_\ell)\in V(G)^\ell, \bar
w=(w_1,\ldots,w_\ell)\in V(H)^\ell$ we have $\atp(G,\bar
v)=\atp(H,\bar w)$ if and only if the mapping $v_i\mapsto w_i$ is an
isomorphism from $G[v_1,\ldots,v_k]$ to $H[w_1,\ldots,w_k]$. Formally, we
can view $\atp(G,v_1,\ldots,v_k)$ as a pair of Boolean $(k\times
k)$-matrices, one describing equalities between the $v_i$, and one
describing adjacencies. If $G$ is a colored graph, we need one matrix
for each arc color and one diagonal matrix for each vertex color.

Then $k$-WL computes a sequence of colorings $\WLIt kiG$ of $V(G)^k$ as
follows. Initially, each tuple is colored by its atomic type, that
is,
$
\WLItVe k0G{\bar v}\coloneqq\atp(G,\bar v).
$
For $i > 0$, we define
\[
\WLItVe kiG{\bar v} \coloneqq \Big(\WLItVe k{i-1}G{\bar
  v},\mathcal{M}^k_i(\bar v)\Big),
\]
where
\[
  \mathcal{M}^k_i(\bar v) \coloneqq
  \Biglmulti
  \big(\atp(G,\bar v v),\WLItVe k{i-1}G{\bar v[v/1]},
  %\WLItVe k{i-1}G{\bar v[v/2]},
  \ldots,
  \WLItVe k{i-1}G{\bar v[v/k]}
  \big)
  \Bigmid
  v\in V(G)\Bigrmulti.
\]
Here, for $\bar v=(v_1,\ldots,v_k)\in V(G)$ and $v\in
V(G)$, by $\bar vv$ we denote the $(k+1)$-tuple $(v_1,\ldots,v_k,v)$
and, for $i\in[k]$, by $\bar v[v/i]$ the $k$-tuple
$(v_1,\ldots,v_{i-1},v,v_{i+1},\ldots,v_k)$.
Then there is an $i_\infty < n^k$
such that $\WLIt k{i_\infty}G\equiv\WLIt k{i_\infty+1}G$, and we call
$\WL kG\coloneqq\WLIt k{i_\infty}G$ the \defword{$k$-stable coloring}.
The $k$-stable coloring of an $n$-vertex graph can be computed  in
time $\CO(n^k\log n)$ \cite{immlan90}.

It is easy to see that $1$-WL is essentially just the Color Refinement algorithm,
that is, for all graphs $G,H$, vertices $v\in V(G),w\in V(H)$, and
$i\ge 0$ it holds that
\begin{equation}
  \label{eq:wl1}
  \ColRefItVe iGv=\ColRefItVe iHw\iff\WLItVe 1iGv=\WLItVe 1iHw.
\end{equation}
In the following, we drop the distinction between the Color Refinement algorithm and $1$-WL.

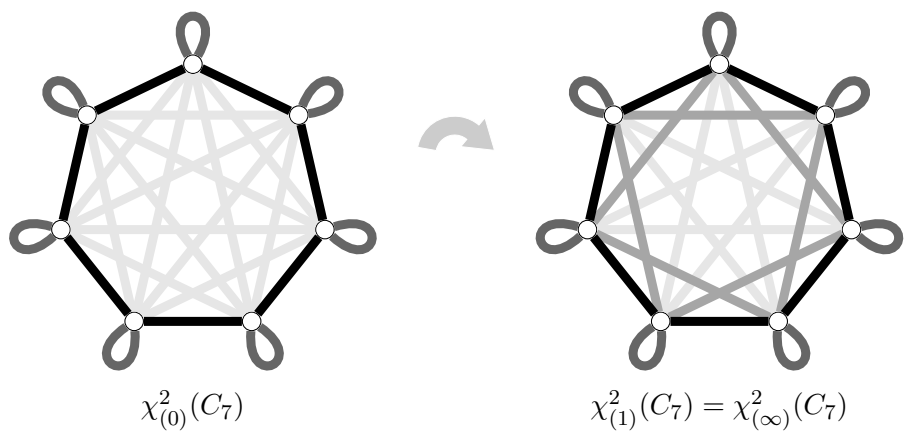
\begin{figure}
  \centering
    \begin{tikzpicture}
   %    \begin{scope}
   %      \foreach \i in {0,...,6}
   %      {
   %        \node[vertex] (v\i) at (90+\i*360/7:2cm) {};
   %      }

   %   \draw[black,line width=0.6mm] (v0) edge (v1)  (v1) edge (v2)
   %   (v2) edge (v3)  (v3) edge (v4)  (v4) edge (v5)  (v5) edge (v6)
   %   (v6) edge (v0);
   % \end{scope}

   \begin{scope}[xshift=0cm]
        \foreach \i in {0,...,6}
        {
          \node[emptyvertex] (v\i) at (90+\i*360/7:1.8cm) {};
          \draw (v\i) edge[black!60,line width=1.2mm,out=120+\i*360/7,in=60+\i*360/7,looseness=16] (v\i);
        }

     \draw[black!10,line width=1.2mm] (v0) edge (v2) edge (v3) edge
     (v4) edge (v5) (v1) edge (v3) edge (v4) edge (v5) edge (v6) (v2)
     edge (v4) edge (v5) edge (v6) (v3) edge (v5) edge (v6) (v4) edge (v6);

     \draw[black,line width=1.2mm] (v0) edge (v1)  (v1) edge (v2)
     (v2) edge (v3)  (v3) edge (v4)  (v4) edge (v5)  (v5) edge (v6)
     (v6) edge (v0);

     \node at (0,-2.8) {$\WLIt20{C_7}$};
   \end{scope}
   
   \node at (3.5,0.9) {
   \scalebox{0.5}{
   \begin{tikzpicture}[rotate=-130]
    \draw[white,fill=black!20] (-60:1cm) arc (-60:-180:1cm) -- (-180:1.25cm) -- (-.75,.75) -- (-180:.25cm) -- (-180:.5cm) arc (-180:-60:.5cm)--cycle;
   \end{tikzpicture}
   }
   };
   
   %\draw[black!20,->,line width=1.5mm] (2.6,0.3) to[bend left] (4.4,0.3);
   
   \begin{scope}[xshift=7cm]
        \foreach \i in {0,...,6}
        {
          \node[emptyvertex] (v\i) at (90+\i*360/7:1.8cm) {};
          \draw (v\i) edge[black!60,line width=1.2mm,out=120+\i*360/7,in=60+\i*360/7,looseness=16] (v\i);
        }

     \draw[black!10,line width=1.2mm] (v0) edge (v3) edge
     (v4) (v1) edge (v4) edge (v5) (v2) edge (v5) edge (v6) (v3) edge (v6);

          \draw[black!35,line width=1.2mm] (v0) edge (v2) edge (v5) (v1) edge (v3) edge (v6) (v2)
     edge (v4)  (v3) edge (v5) (v4) edge (v6);

     \draw[black,line width=1.2mm] (v0) edge (v1)  (v1) edge (v2)
     (v2) edge (v3)  (v3) edge (v4)  (v4) edge (v5)  (v5) edge (v6)
     (v6) edge (v0);

     \node at (0,-2.8) {$\WLIt21{C_7}=\WL2{C_7}$};
   \end{scope}
    \end{tikzpicture}
  \caption{The iterations of $2$-WL on a cycle $C_7$}
  \label{fig:wl-cycle}
\end{figure}

\begin{example}
The sequence of colorings of $2$-WL on a cycle $C_7$ of length $7$ is shown in Figure
\ref{fig:wl-cycle}. In this example, the coloring is symmetric,
that is, $\WLItVe 2i{C_7}{v,v'}=\WLItVe 2i{C_7}{v',v}$ for all $i\ge
0$ and $v,v'\in V(C_7)$. This allows us to show the colors as undirected edges.
\end{example}

\subsection{The Power of WL and the WL Dimension}
To use the Weisfeiler-Leman algorithm as an isomorphism
test, we can compare the color patterns of two given graphs. We say
that $k$-WL \defword{distinguishes} graphs $G$ and $H$ if there is a
color $c\in\operatorname{rg}(\WL kG)\cup\operatorname{rg}(\WL kH)$ such that
\[
  \big|\big\{ \bar v\in V(G)^k\bigmid \WLVe kG{\bar v}=c\big\}\big|\neq
  \big|\big\{ \bar w\in V(H)^k\bigmid \WLVe kH{\bar w}=c\big\}\big|.
\]
Note that $k$-WL is an \defword{incomplete} isomorphism test: if it
distinguishes two graphs, we know that they are non-isomorphic, but as
we will see, there are non-isomorphic graphs that $k$-WL does not
distinguish. Let us call two graphs that are not distinguished by the
$k$-WL \defword{$k$-indistinguishable}.

\begin{example}
  Any two $d$-regular graphs are $1$-indistinguishable. The simplest
  example are  a cycle of length $6$ and the disjoint union of two triangles.
  
  Similarly, any two strongly regular graphs with the same parameters
  are $2$-indistinguishable. Figure~\ref{fig:sr-graphs} shows two
  non-isomorphic strongly regular graphs with parameters $(16,6,2,2)$,
  which means that they have 16 vertices, each vertex has $6$
  neighbors, each pair of adjacent vertices has two common
  neighbors, and each pair of non-adjacent vertices has 2 common neighbors. 
  The first of the graphs is the line graph of the complete bipartite graph $K_{4,4}$.
  The second graph is the Shrikhande graph, which
  can be viewed as the undirected graph underlying the Cayley graph
  for the group $\mathbb Z_4\times\mathbb Z_4$ with generators
  $(0,1),(1,0),(1,1)$. To see that they are
  non-isomorphic, note that the neighborhood of each vertex in the
  Shrikhande graph is a cycle of length $6$, and the neighborhood of
  the each vertex in the line graph of $K_{4,4}$ is the disjoint union
  of two triangles.
\end{example}

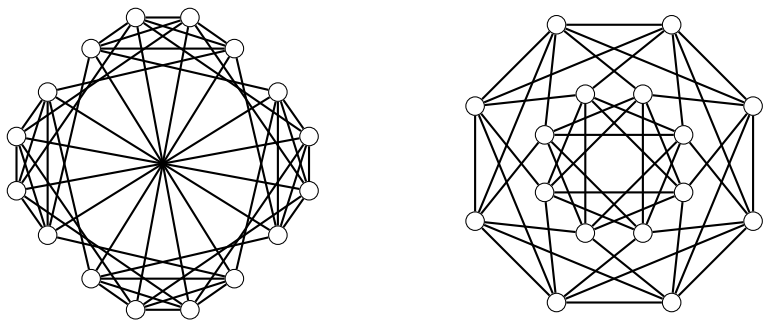
\begin{figure}
 \centering
 \begin{tikzpicture}
  \foreach \i in {0,...,3}{
   \foreach \j in {0,...,3}{
    \node[emptyvertex] (v\i\j) at ($(90*\i:0.8) + (90*\i+35*\j-52.5:1.2)$) {};
   }
  }
  
  \foreach \i in {0,...,3}{
   \foreach \j/\k in {0/1,0/2,0/3,1/2,1/3,2/3}
   \path[draw,thick] (v\i\j) edge (v\i\k);
  }
  
  \foreach \i in {0,...,3}{
   \foreach \j/\k in {0/1,0/2,0/3,1/2,1/3,2/3}
   \path[draw,thick] (v\j\i) edge (v\k\i);
  }
  
  \foreach \i in {0,...,7}{
   \node[emptyvertex] (u\i) at ($(6,0)+(22.5+45*\i:1)$) {};
   \node[emptyvertex] (v\i) at ($(6,0)+(22.5+45*\i:2)$) {};
  }
  \foreach \i/\j in {0/2,0/3,0/5,0/6,1/3,1/4,1/6,1/7,2/4,2/5,2/7,3/5,3/6,4/6,4/7,5/7}{
   \path[draw,thick] (u\i) edge (u\j);
  }
  \foreach \i/\j in {0/1,0/7,1/0,1/2,2/1,2/3,3/2,3/4,4/3,4/5,5/4,5/6,6/5,6/7,7/0,7/6}{
   \path[draw,thick] (v\i) edge (u\j);
  }
  \foreach \i/\j in {0/1,0/2,0/6,0/7,1/2,1/3,1/7,2/3,2/4,3/4,3/5,4/5,4/6,5/6,5/7,6/7}{
   \path[draw,thick] (v\i) edge (v\j);
  }
 \end{tikzpicture}
 \caption{Two non-isomorphic strongly regular graphs with parameters $(16,6,2,2)$: the line graph of $K_{4,4}$ (left) and the Shrikhande Graph (right).}
 \label{fig:sr-graphs}
\end{figure}

It is not easy to find two non-isomorphic graphs that are $k$-indistinguishable for some $k \geq 3$.
In a seminal paper, Cai, Fürer, and Immerman constructed such graphs for each $k$.

\begin{theorem}[Cai, Fürer and Immerman~\cite{CaiFI92}]
 \label{theo:cfi}
 For every $k\ge 1$ there exist non-isomorphic $3$-regular graphs
 $G_k,H_k$ of order $|G_k|=|H_k|=\CO(k)$ that are $k$-in\-dis\-tin\-guish\-able.
\end{theorem}

The Weisfeiler-Leman algorithm colors tuples of vertices based on the
isomorphism type of these tuples and the overlap with other tuples.
To distinguish the so called \defword{CFI graphs} $G_k,H_k$ of
Theorem~\ref{theo:cfi}, we need to color tuples of length linear in
the order of the graphs. As the number of such tuples is
$\Omega(k^k)$, this is prohibitive both in terms of running time and
memory space. However, a closer analysis reveals that it suffices to
color only some long tuples, to be precise, an isomorphism-invariant set of
tuples of size polynomial in $k$ (this follows from an argument due to
Dawar, Richerby and Rossman~\cite{dawricros06}). In
\cite{groschwewie21}, the first author jointly with Schweitzer and
Wiebking proposed an algorithmic framework called \textsc{DeepWL} for
combinatorial ``coloring'' algorithms that
allows it to selectively color tuples within isomorphism-invariant
families and thereby color tuples of length linear in the input
graphs. \textsc{DeepWL} can distinguish the CFI graphs in polynomial time.

While Theorem~\ref{theo:cfi} clearly shows its limitations, the
Weisfeiler-Leman algorithm still remains quite powerful. In
particular, it does yield complete isomorphism tests for many natural
graph classes. Let us say that $k$-WL \defword{identifies} a graph $G$ if
it distinguishes $G$ from all graphs $H$ that are not isomorphic to
$G$. It is long known \cite{baberdsel80} that even $1$-WL (that is, the Color Refinement algorithm)
identifies almost all graphs, in the sense that in the limit, the fraction of
$n$-vertex graphs identified by $1$-WL is $1$ as $n$ goes to
infinity. Similarly, $2$-WL identifies almost all $d$-regular graphs
\cite{bol82}.

The \defword{WL dimension} of a graph $G$ is the least $k\ge 1$ such that
$k$-WL identifies $G$. Trivially, the WL dimension of an $n$-vertex
graph is at most $n$ (in fact, $n-1$ for $n\ge 2$).
The \defword{WL dimension} of a class $\CC$ of graphs is the maximum of
the WL dimensions of the graphs in $\CC$ if this maximum exists, or
$\infty$ otherwise.
It is easy to
see that the WL dimension of the class of forests is $1$. Building on
earlier results for graphs of bounded genus and bounded tree width
\cite{gro98a,gro00,gromar99}, the first author proved the following
far-reaching result.

\begin{theorem}[Grohe~\cite{gro17}]
 \label{theo:x-minor}
 All graph classes excluding a fixed graph as a minor have bounded WL dimension.
\end{theorem}

Explicit bounds on the WL dimension are known for specific graph
classes. Most notably, Kiefer, Ponomarenko, and Schweitzer
\cite{KieferPS19} proved that the WL dimension of planar graphs is at
most $3$. The example of two triangles vs a cycle of length $6$ shows
that it is at least $2$. More generally, the WL dimension of the class
of all graphs embeddable in a surface of Euler genus $g$ is at most
$4g+3$ \cite{grokie19}; a lower bound linear in $g$ follows from
Theorem~\ref{theo:cfi}. The second author jointly with Kiefer
\cite{kieneu19} proved that the WL dimension of the class of graphs of
tree width $k$ is at least $\lceil k/2\rceil-3$ and at most $k$.

There are also natural graph classes that do not exclude a fixed graph as a
minor (and that are therefore not covered by
Theorem~\ref{theo:x-minor}). The WL dimension of the class of interval
graphs is $2$ \cite{evdpontin00} and the WL dimension of
the class of graphs of rank width at most $k$ is at most $3k+4$
\cite{groneu19}. Contrasting all these results,
Theorem~\ref{theo:cfi} shows that the WL dimension of the class of
3-regular graphs is $\infty$, even though isomorphism of 3-regular
graphs can be decided in polynomial time \cite{Luks82} (cf.~Theorem~\ref{theo:bounded-degree}).
Much more on this topic can found in Kiefer's recent survey \cite{kie20}.

\subsection{Characterisations of WL}
The relevance of the Weisfeiler-Leman algorithm is not restricted to GI.
It has found applications in linear
optimisation~\cite{grokermla+14} and probabilistic
inferences~\cite{ahmkermlanat13}, and because of relations to graph
kernels \cite{sheschlee+11} and graph neural networks
\cite{morritfey+19,xuhulesjeg19}, received attention in machine
learning recently. On a theoretical level, connections to other areas
are induced by several seemingly unrelated characterizations of
$k$-indistinguishability in terms of logic, combinatorics, and algebra.

We start with a logical characterization going back to
\cite{CaiFI92,immlan90}. Let $\LC$ be the extension of first-order predicate logic using
counting quantifiers of the form $\exists^{\ge \ell}x$ for $\ell \ge 1$,
where $\exists^{\ge \ell}x\,\varphi(x)$ means that there are at least $\ell$
elements $x$ satisfying $\varphi$. By $\LC^k$ we denote the fragment of
$\LC$ consisting of formulas with at most $k$ variables. For example,
the following $\LC^2$-sentence $\varphi$ says that all vertices of a graph
have fewer than $3$ neighbors of degree at least $10$:
\begin{equation*}
  \varphi\coloneqq\forall x\neg\exists^{\ge 3}y\big(
  E(x,y)\wedge \exists^{\ge 10} x\,E(y,x)\big).
\end{equation*}
Note that we reuse variable $x$ for the innermost quantification to keep the number of distinct variables down to two.

\begin{theorem}[Cai, Fürer and Immerman~\cite{CaiFI92}]
  Let $k\ge 1$. Then two graphs are $k$-indistinguishable if and only if
  they are $\LC^{k+1}$-equivalent, that is, they satisfy the same
  sentences of the logic $\LC^{k+1}$.
\end{theorem}

This logical characterization of $k$-indistinguishability plays a
crucial role in the proof of Theorem~\ref{theo:x-minor} and, via a game
characterization of $\LC^k$-equivalence, in the proof of
Theorem~\ref{theo:cfi}.

The next characterization we will discuss is combinatorial. Recall
that a graph homomorphism is a mapping that preserves adjacency, and
let $\hom(G,H)$ denote the number of homomorphism from $G$ to $H$. An
old theorem due to Lovász \cite{lov67} states that two graphs $G,H$ are
isomorphic if and only if $\hom(F,G)=\hom(F,H)$ for all graphs
$F$. If, instead of all graphs $F$, we only count homomorphisms from
graphs of tree width $k$, we obtain a characterization of
$k$-indistinguishability.

\begin{theorem}[Dvorák~\cite{dvo10}]
 Let $k\ge 1$. Then two graphs
 $G,H$ are $k$-in\-dis\-tin\-guish\-able if and only if\/ $\hom(F,G)=\hom(F,H)$
 for all graphs $F$ of tree width at most $k$.
\end{theorem}

Maybe the most surprising characterization of $k$-indistinguishability
is algebraic. For simplicity, we focus on $k=1$ here. Let $G,H$ be graphs with vertex sets
$V \coloneqq V(G),W \coloneqq V(H)$, and let $A\in\mathbb R^{V\times V},B\in\mathbb R^{W\times
  W}$ be their adjacency matrices. Then $G$ and $H$ are isomorphic if
and only if there is a permutation matrix $X$ such that $X^{-1}AX=B$,
or equivalently, $AX=XB$.
We can write this as a system of linear (in)equalities:
\begin{align}
  \label{eq:lp2}
  \sum_{v'\in V}A_{vv'}X_{v'w}&=\sum_{w'\in W}X_{vw'}B_{w'w},\\
  \label{eq:lp3}
  \sum_{w'\in W}X_{vw'}&=\sum_{v'\in V}X_{v'w}=1,\\
  \label{eq:lp4}
  X_{vw}&\ge 0,&\text{for all }v\in V,w\in W.
\end{align}
Integer solutions to this system are precisely permutation matrices
$X=(X_{vw})$ 
satisfying $AX=XB$. Thus the system \eqref{eq:lp2}--\eqref{eq:lp4} has
an integer solution if and only if $G$ and $H$ are isomorphic. Now let
us drop the integrality constraints and consider rational solutions to
the system, which can be viewed as doubly stochastic matrices $X$
satisfying $AX=XB$. Let us call such solutions \defword{fractional
  isomorphisms}. 

\begin{theorem}[Tinhofer~\cite{tin91}]
 Two graphs $G,H$ are $1$-distinguishable if and only if they are fractionally isomorphic,
 that is, the system \eqref{eq:lp2}--\eqref{eq:lp4} has a rational solution.
\end{theorem}

Atserias and Maneva~\cite{atsman13} showed that by 
considering the Sherali-Adams hierarchy over the integer linear
program \eqref{eq:lp2}--\eqref{eq:lp4}, we obtain a characterization
of $k$-indistinguishability. To be precise, solvability of the
level-$k$ Sherali-Adams (in)equalities yields an equivalence relation
strictly between $k$-indistinguishability and
$(k+1)$-indistinguishability. By combining the equations of levels
$(k-1)$ and $k$, we obtain an exact characterization of
$k$-indistinguishability \cite{groott15,mal14}. A similar, albeit
looser correspondence to $k$-indistinguishability has been obtained
for the Lassere hierarchy of semi-definite relaxations of the integer
linear program \eqref{eq:lp2}--\eqref{eq:lp4} \cite{atsoch18,odowriwu+14} and
in terms of polynomial equations and the so-called polynomial calculus
\cite{atsoch18,bergro15,gragropagpak19}.

\section{The Group-Theoretic Graph Isomorphism Machinery}
\label{sec:group}

In this section we present recent advances on the group-theoretic graph isomorphism machinery.
After giving some basics on permutation groups we first discuss Luks's algorithm for testing isomorphism of graphs of bounded degree \cite{Luks82}.
This algorithm forms the foundation of the group-theoretic graph isomorphism machinery.
Then, we give a very brief overview on Babai's quasi-polynomial-time isomorphism test \cite{Babai16} which also builds on Luks's algorithm and attacks the obstacle cases where the recursion performed by Luks's algorithm does not give the desired running time.
Finally, we present recent extensions of Babai's techniques resulting
in more efficient isomorphism tests for bounded-degree graphs
\cite{GroheNS18} and hypergraphs \cite{Neuen20}.

\subsection{Basics}

We introduce the group-theoretic notions required in this work.
For a general background on group theory we refer to \cite{Rotman99}; background on permutation groups can be found in \cite{DixonM96}.

\paragraph{Permutation Groups.}

A \defword{permutation group} acting on a set $\Omega$ is a subgroup $\Gamma \leq \Sym(\Omega)$ of the symmetric group.
The size of the permutation domain $\Omega$ is called the \defword{degree} of $\Gamma$.
Throughout this section, the degree is denoted by $n \coloneqq |\Omega|$.
If $\Omega = [n]$, then we also write $S_n$ instead of $\Sym(\Omega)$.
The \defword{alternating group} on a set $\Omega$ is denoted by $\Alt(\Omega)$.
As before, we write $A_n$ instead of $\Alt(\Omega)$ if $\Omega = [n]$.
For $\gamma \in \Gamma$ and $\alpha \in \Omega$ we denote by $\alpha^{\gamma}$ the image of $\alpha$ under the permutation $\gamma$.
The set $\alpha^{\Gamma} \coloneqq \{\alpha^{\gamma} \mid \gamma \in \Gamma\}$ is the \defword{orbit} of $\alpha$.
The group $\Gamma$ is \defword{transitive} if $\alpha^\Gamma = \Omega$ for some (and therefore every) $\alpha \in \Omega$.

For $\alpha \in \Omega$ the group
$\Gamma_\alpha \coloneqq \{\gamma \in \Gamma \mid \alpha^{\gamma} = \alpha\}
\leq \Gamma$ is the \defword{stabilizer} of $\alpha$ in $\Gamma$.  The
group $\Gamma$ is \defword{semi-regular} if $\Gamma_\alpha = \{\id\}$ for
every $\alpha \in \Omega$ (where $\id$ denotes the identity element of
the group).  The \defword{pointwise stabilizer} of a set
$A \subseteq \Omega$ is the subgroup
$\Gamma_{(A)} \coloneqq \{\gamma \in \Gamma \mid\forall \alpha \in A\colon
\alpha^{\gamma}= \alpha \}$.  For $A \subseteq \Omega$ and
$\gamma \in \Gamma$ let
$A^{\gamma} \coloneqq \{\alpha^{\gamma} \mid \alpha \in A\}$.  The set $A$ is
\defword{$\Gamma$-invariant} if $A^{\gamma} = A$ for all
$\gamma \in \Gamma$.  The \defword{setwise stabilizer} of a set
$A \subseteq \Omega$ is the subgroup
$\Gamma_{A} \coloneqq \{\gamma \in \Gamma \mid A^{\gamma} = A\}$.

For $A \subseteq \Omega$ and a bijection $\theta\colon \Omega \rightarrow \Omega'$ we denote by $\theta[A]$ the restriction of $\theta$ to the domain $A$.
For a $\Gamma$-invariant set $A \subseteq \Omega$, we denote by $\Gamma[A] \coloneqq \{\gamma[A] \mid \gamma \in \Gamma\}$ the induced action of $\Gamma$ on $A$, i.e., the group obtained from $\Gamma$ by restricting all permutations to $A$.

Let $\Gamma \leq \Sym(\Omega)$ be a transitive group.
A \defword{block} of $\Gamma$ is a nonempty subset $B \subseteq \Omega$ such that $B^{\gamma} = B$ or $B^{\gamma} \cap B = \emptyset$ for all $\gamma \in \Gamma$.
The trivial blocks are $\Omega$ and the singletons $\{\alpha\}$ for $\alpha \in \Omega$.
The group $\Gamma$ is called \defword{primitive} if there are no non-trivial blocks.
If $B \subseteq \Omega$ is a block of $\Gamma$ then $\FB \coloneqq \{B^{\gamma} \mid \gamma \in \Gamma\}$ builds a \defword{block system} of $\Gamma$.
Note that $\FB$ is an equipartition of $\Omega$.
The group $\Gamma_{(\FB)} \coloneqq \{\gamma \in \Gamma \mid \forall B \in \FB\colon B^{\gamma} = B\}$ denotes the subgroup stabilizing each block $B \in \FB$ setwise.
Moreover, the natural action of $\Gamma$ on the block system $\FB$ is denoted by $\Gamma[\FB] \leq \Sym(\FB)$.
Let $\FB,\FB'$ be two partitions of $\Omega$.
We say that $\FB$ \defword{refines} $\FB'$, denoted $\FB \preceq \FB'$, if for every $B \in \FB$ there is some $B' \in \FB'$ such that $B \subseteq B'$.
If additionally $\FB \neq \FB'$ we write $\FB \prec \FB'$.
A block system $\FB$ is \defword{minimal} if there is no non-trivial block system $\FB'$ such that $\FB \prec \FB'$.
A block system $\FB$ is minimal if and only if $\Gamma[\FB]$ is primitive.

Let $\Gamma \leq \Sym(\Omega)$ and $\Gamma' \leq \Sym(\Omega')$.
A \defword{homomorphism} is a mapping $\varphi\colon \Gamma \rightarrow \Gamma'$ such that $\varphi(\gamma)\varphi(\delta) = \varphi(\gamma\delta)$ for all $\gamma,\delta \in \Gamma$.
For $\gamma \in \Gamma$ we denote by $\gamma^{\varphi}$ the $\varphi$-image of $\gamma$.
Similarly, for $\Delta \leq \Gamma$ we denote by $\Delta^{\varphi}$ the $\varphi$-image of $\Delta$ (note that $\Delta^{\varphi}$ is a subgroup of $\Gamma'$).

\paragraph{Algorithms for Permutation Groups.}

Next, we review some basic facts about algorithms for permutation groups.
For detailed information we refer to \cite{Seress03}.

In order to handle permutation groups computationally it is essential
to represent groups in a compact way.  Indeed, the size of a
permutation group is typically exponential in the degree of the group
which means it is not possible to store the whole group in memory.  In
order to allow for efficient computation, permutation groups are
represented by generating sets.  By Lagrange's Theorem, for each
permutation group $\Gamma \leq \Sym(\Omega)$, there is a generating
set of size $\log|\Gamma| \leq n \log n$ (recall that
$n \coloneqq |\Omega|$ denotes the size of the permutation domain).
Actually, most algorithms are based on so-called \defword{strong
  generating sets}, which can be chosen of size quadratic in the
degree of the group and can be computed in polynomial time given an
arbitrary generating set (see, e.g., \cite{Seress03}).

\begin{theorem}[cf.\ \cite{Seress03}] 
 \label{thm:permutation-group-library}
 Let $\Gamma \leq \Sym(\Omega)$ and let $S$ be a generating set for $\Gamma$.
 Then the following tasks can be performed in time polynomial in $n$ and $|S|$:
 \begin{enumerate}
  \item compute the order of $\Gamma$,
  \item given $\gamma \in \Sym(\Omega)$, test whether $\gamma \in \Gamma$,
  \item compute the orbits of $\Gamma$,
  \item compute a minimal block system of $\Gamma$ (if $\Gamma$ is transitive), and
  \item given $A \subseteq \Omega$, compute a generating set for $\Gamma_{(A)}$.
 \end{enumerate}
 Let $\Delta \leq \Sym(\Omega')$ be a second group of degree $n' = |\Omega'|$.
 Also let $\varphi\colon \Gamma \rightarrow \Delta$ be a homomorphism given by a list of images for $s \in S$.
 Then the following tasks can be solved in time polynomial in $n$, $n'$ and $|S|$:
 \begin{enumerate}
  \setcounter{enumi}{5}
  \item compute a generating set for $\ker(\varphi) \coloneqq \{\gamma \in \Gamma \mid \gamma^{\varphi} = \id\}$, and
  \item given $\delta \in \Delta$, find $\gamma \in \Gamma$ such that $\gamma^{\varphi} = \delta$ (if it exists).
 \end{enumerate}
\end{theorem}

Observe that, by always maintaining generating sets of size at most quadratic in the degree,
we can concatenate a polynomial number of subroutines from the theorem while keeping the polynomial-time bound.
For the remainder of this work, we will typically ignore the role of generating sets and simply refer to groups being the input/output of an algorithm. 
This always means the algorithm performs computations on generating sets of size polynomial in the degree of the group.

\subsection{Luks's Algorithm}
\label{subsec:luks}

\paragraph{The String Isomorphism Problem.}

In order to apply group-theoretic techniques to the Graph Isomorphism
Problem, Luks~\cite{Luks82} introduced a more general problem that
allows to build recursive algorithms along the structure of the
permutation groups involved.  Here, we follow the notation and
terminology used by Babai \cite{Babai16} for describing his
quasi-polynomial-time algorithm for the Graph Isomorphism Problem that
also employs this recursive strategy.

A \defword{string} is a mapping
$\Fx\colon\Omega\rightarrow\Sigma$ where $\Omega$ is a finite
set and $\Sigma$ is also a finite set called the \defword{alphabet}.
Let $\gamma \in \Sym(\Omega)$ be a permutation.
The permutation $\gamma$ can be applied to the string $\Fx$ by defining
\begin{equation}
 \Fx^{\gamma}\colon\Omega\rightarrow\Sigma\colon\alpha\mapsto\Fx\left(\alpha^{\gamma^{-1}}\right).
\end{equation}
Let $\Fy\colon\Omega\rightarrow\Sigma$ be a second string.
The permutation $\gamma$ is an isomorphism from $\Fx$ to $\Fy$, denoted $\gamma\colon \Fx \cong \Fy$, if $\Fx^{\gamma} = \Fy$.
Let $\Gamma \leq \Sym(\Omega)$.
A \defword{$\Gamma$-isomorphism} from $\Fx$ to $\Fy$ is a permutation $\gamma \in \Gamma$ such that $\gamma\colon \Fx \cong \Fy$.
The strings $\Fx$ and $\Fy$ are \defword{$\Gamma$-isomorphic}, denoted $\Fx \cong_\Gamma \Fy$, if there is a $\Gamma$-isomorphism from $\Fx$ to $\Fy$.
The \defword{String Isomorphism Problem} asks, given two strings $\Fx,\Fy\colon\Omega\rightarrow\Sigma$ and a generating set for a group $\Gamma \leq \Sym(\Omega)$, whether $\Fx$ and $\Fy$ are $\Gamma$-isomorphic.
The set of $\Gamma$-isomorphisms from $\Fx$ to $\Fy$ is denoted by
\begin{equation}
 \Iso_\Gamma(\Fx,\Fy) \coloneqq \{\gamma \in \Gamma \mid \Fx^{\gamma} = \Fy\}.
\end{equation}
The set of \defword{$\Gamma$-automorphisms} of $\Fx$ is $\Aut_\Gamma(\Fx) \coloneqq \Iso_\Gamma(\Fx,\Fx)$.
Observe that $\Aut_\Gamma(\Fx)$ is a subgroup of $\Gamma$ and $\Iso_\Gamma(\Fx,\Fy) = \Aut_\Gamma(\Fx)\gamma \coloneqq \{\delta\gamma \mid \delta \in \Aut_\Gamma(\Fx)\}$ for an arbitrary $\gamma \in \Iso_\Gamma(\Fx,\Fy)$.

\begin{theorem}
 \label{thm:reduction-graph-isomorphism-to-string-isomorphism}
 There is a polynomial-time many-one reduction from the Graph Isomorphism Problem to the String Isomorphism Problem.
\end{theorem}

\begin{figure}
 \centering
 \begin{tikzpicture}
  \node[emptyvertex,label={135:$1$}] (v1) at (0,1.6) {};
  \node[emptyvertex,label={45:$2$}] (v2) at (1.6,1.6) {};
  \node[emptyvertex,label={225:$3$}] (v3) at (0,0) {};
  \node[emptyvertex,label={315:$4$}] (v4) at (1.6,0) {};
  \node at (-0.8,0.8) {$G$};
  
  \node[emptyvertex,label={135:$1$}] (w1) at (5.6,1.6) {};
  \node[emptyvertex,label={45:$2$}] (w2) at (7.2,1.6) {};
  \node[emptyvertex,label={225:$3$}] (w3) at (5.6,0) {};
  \node[emptyvertex,label={315:$4$}] (w4) at (7.2,0) {};
  \node at (4.8,0.8) {$H$};
  
  \path[draw,thick] (v1) edge (v2);
  \path[draw,thick] (v1) edge (v3);
  \path[draw,thick] (v2) edge (v4);
  \path[draw,thick] (v3) edge (v4);
  
  \path[draw,thick] (w1) edge (w2);
  \path[draw,thick] (w1) edge (w4);
  \path[draw,thick] (w2) edge (w3);
  \path[draw,thick] (w3) edge (w4);
  
  \node at (3.6,0.8) {$\cong$};
  
  \node[rotate=-90] at (0.8,-0.75) {{\large $\rightsquigarrow$}};
  \node[rotate=-90] at (6.4,-0.75) {{\large $\rightsquigarrow$}};
  
  \node at (3.6,-1.8) {$\cong_{\Gamma}$};
  
  \draw (-1,-1.5) -- (2.6,-1.5);
  \draw (-1,-2.1) -- (2.6,-2.1);
  \foreach \i in {0,1,2,3,4,5,6}{
   \draw (-1+0.6*\i,-1.5) -- (-1+0.6*\i,-2.1);
  }
  \foreach[count=\i] \e/\a in {12/1,13/1,14/0,23/0,24/1,34/1}{
   \node at (-1.3+0.6*\i,-1.8) {$\a$};
   \node[rotate=90] at (-1.3+0.6*\i,-2.4) {{\color{gray}$\e$}};
  }
  \node at (-1.3,-1.8) {$\Fx$};
  
  \draw (4.6,-1.5) -- (8.2,-1.5);
  \draw (4.6,-2.1) -- (8.2,-2.1);
  \foreach \i in {0,1,2,3,4,5,6}{
   \draw (4.6+0.6*\i,-1.5) -- (4.6+0.6*\i,-2.1);
  }
  \foreach[count=\i] \e/\a in {12/1,13/0,14/1,23/1,24/0,34/1}{
   \node at (4.3+0.6*\i,-1.8) {$\a$};
   \node[rotate=90] at (4.3+0.6*\i,-2.4) {{\color{gray}$\e$}};
  }
  \node at (4.3,-1.8) {$\Fy$};
  
 \end{tikzpicture}
 \caption{Reduction from the Graph Isomorphism Problem to the String Isomorphism Problem. The group $\Gamma$ is the natural induced action of $S_4$ on $\binom{[4]}{2}$.}
 \label{fig:gi-to-si}
\end{figure}
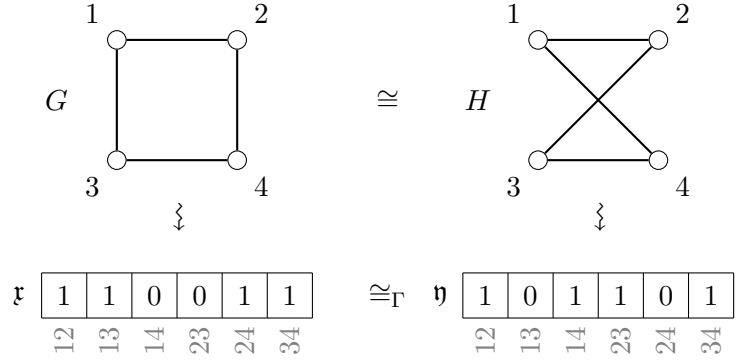

\begin{proof}
 Let $G$ and $H$ be two graphs and assume without loss of generality $V(G) = V(H) =: V$.
 Let $\Omega \coloneqq \binom{V}{2}$ and $\Gamma \leq \Sym(\Omega)$ be the natural action of $\Sym(V)$ on the two-element subsets of $V$.
 Let $\Fx \colon \Omega \rightarrow \{0,1\}$ with $\Fx(vw) = 1$ if and only if $vw \in E(G)$.
 Similarly define $\Fy \colon \Omega \rightarrow \{0,1\}$ with $\Fy(vw) = 1$ if and only if $vw \in E(H)$ (see Figure \ref{fig:gi-to-si}).
 Let $\varphi\colon\Sym(V) \rightarrow \Gamma$ be the natural homomorphism.
 Then $\gamma\colon G \cong H$ if and only if $\gamma^{\varphi} \colon \Fx \cong \Fy$.
\end{proof}

The main advantage of the String Isomorphism Problem is that it naturally allows for algorithmic approaches based on group-theoretic techniques.

\paragraph{Recursion Mechanisms.}

The foundation for the group-theoretic approaches lays in two recursion mechanisms first exploited by Luks \cite{Luks82}.
As before, let $\Fx,\Fy\colon \Omega \rightarrow \Sigma$ be two strings.
For a set of permutations $K \subseteq \Sym(\Omega)$ and a \defword{window} $W \subseteq \Omega$ define
\begin{equation}
 \Iso_K^{W}(\Fx,\Fy) \coloneqq \{\gamma \in K \mid \forall \alpha \in W\colon \Fx(\alpha) = \Fy(\alpha^{\gamma})\}.
\end{equation}
In the following, the set $K$ is always a \defword{coset}, i.e., $K = \Gamma\gamma \coloneqq \{\gamma'\gamma \mid \gamma'\in\Gamma\}$ for some group $\Gamma \leq \Sym(\Omega)$ and a permutation $\gamma \in \Sym(\Omega)$, and the set $W$ is $\Gamma$-invariant.
In this case it can be shown that $\Iso_K^{W}(\Fx,\Fy)$ is either empty or a coset of the group $\Aut_\Gamma^{W}(\Fx) \coloneqq \Iso_\Gamma^{W}(\Fx,\Fx)$.
Hence, the set $\Iso_K^{W}(\Fx,\Fy)$ can be represented by a generating set for $\Aut_\Gamma^{W}(\Fx)$ and a single permutation $\gamma \in \Iso_K^{W}(\Fx,\Fy)$.
Moreover, for $K = \Gamma\gamma$, it holds that
\begin{equation}
 \label{eq:string-isomorphism-alignment}
 \Iso_{\Gamma\gamma}^{W}(\Fx,\Fy) = \Iso_\Gamma^{W}(\Fx,\Fy^{\gamma^{-1}})\gamma.
\end{equation}
Using this identity, it is possible to restrict to the case where $K$ is actually a group.

With these definitions we can now formulate the two recursion mechanisms introduced by Luks \cite{Luks82}.
For the first type of recursion suppose $K = \Gamma \leq \Sym(\Omega)$ is not transitive on $W$ and let $W_1,\dots,W_\ell$ be the orbits of $\Gamma[W]$.
Then the strings can be processed orbit by orbit as follows.
We iterate over all $i \in [\ell]$ and update $K \coloneqq \Iso_K^{W_i}(\Fx,\Fy)$ in each iteration.
In the end, $K = \Iso_\Gamma^{W}(\Fx,\Fy)$.
This recursion mechanism is referred to as \defword{orbit-by-orbit processing}.

The set $\Iso_K^{W_i}(\Fx,\Fy)$ can be computed by making one recursive call to the String Isomorphism Problem over domain size $n_i \coloneqq |W_i|$.
Indeed, using Equation (\ref{eq:string-isomorphism-alignment}), it can be assumed that $K$ is a group and $W_i$ is $K$-invariant.
Then
\begin{equation}
 \Iso_K^{W_i}(\Fx,\Fy) = \{\gamma \in K \mid \gamma[W_i] \in \Iso_{K[W_i]}(\Fx[W_i],\Fy[W_i])\}
\end{equation}
where $\Fx[W_i]$ denotes the induced substring of $\Fx$ on the set $W_i$, i.e., $\Fx[W_i] \colon W_i \rightarrow \Sigma\colon \alpha \mapsto \Fx(\alpha)$ (the string $\Fy[W_i]$ is defined analogously).
So overall, if $\Gamma$ is not transitive, the set $\Iso_\Gamma^{W}(\Fx,\Fy)$ can be computed by making $\ell$ recursive calls over window size $n_i = |W_i|$, $i \in [\ell]$.

For the second recursion mechanism let $\Delta \leq \Gamma$ and let $T$ be a \defword{transversal}\footnote{A set $T \subseteq \Gamma$ is a \defword{transversal} for $\Delta$ in $\Gamma$ if $|T| = |\Gamma|/|\Delta|$ and $\{\Delta\delta \mid \delta \in T\} = \{\Delta\delta \mid \delta \in \Gamma\}$ (such a set always exists).} for $\Delta$ in $\Gamma$.
Then
\begin{equation}
 \Iso_\Gamma^{W}(\Fx,\Fy) = \bigcup_{\delta \in T} \Iso_{\Delta\delta}^{W}(\Fx,\Fy) = \bigcup_{\delta \in T} \Iso_{\Delta}^{W}(\Fx,\Fy^{\delta^{-1}})\delta.
\end{equation}
Luks applied this type of recursion when $\Gamma$ is transitive (on the window $W$), $\FB$ is a minimal block system for $\Gamma$, and $\Delta = \Gamma_{(\FB)}$.
In this case $\Gamma[\FB]$ is a primitive group.
Let $t = |\Gamma[\FB]|$ be the size of a transversal for $\Delta$ in $\Gamma$.
Note that $\Delta$ is not transitive (on the window $W$).
Indeed, each orbit of $\Delta$ has size at most $n/b$ where $b = |\FB|$.
So by combining both types of recursion the computation of $\Iso_\Gamma^{W}(\Fx,\Fy)$ is reduced to $t \cdot b$ many instances of the String Isomorphism Problem over window size $|W|/b$.
This specific combination of types of recursion is referred to as \defword{standard Luks reduction}.
Observe that the time complexity of standard Luks reduction is determined by the size of the primitive group $\Gamma[\FB]$.

\begin{algorithm}
 \caption{Luks's Algorithm: ${\sf StringIso}(\Fx,\Fy,\Gamma,\gamma,W)$}
 \label{alg:luks}
 
 \SetKwInOut{Input}{Input}
 \SetKwInOut{Output}{Output}
 \Input{Strings $\Fx,\Fy\colon \Omega \rightarrow \Sigma$, a group $\Gamma \leq \Sym(\Omega)$, a permutation $\gamma \in \Sym(\Omega)$, and a $\Gamma$-invariant set $W \subseteq \Omega$}
 \Output{$\Iso_{\Gamma\gamma}^W(\Fx,\Fy)$}
 \DontPrintSemicolon
 \BlankLine
 \If{$\gamma \neq 1$}{
  \Return ${\sf StringIso}(\Fx,\Fy^{\gamma^{-1}},\Gamma,1,W)\gamma$\tcc*[r]{{\small Equation (\ref{eq:string-isomorphism-alignment})}}
 }
 \If{$|W| = 1$}{
  \eIf{$\Fx(\alpha) = \Fy(\alpha)$ where $W = \{\alpha\}$}{
   \Return $\Gamma$\;
  }{
   \Return $\emptyset$
  }
 }
 \If{$\Gamma$ is not transitive on $W$}{
  compute orbit $W' \subseteq W$\;
  \Return ${\sf StringIso}(\Fx,\Fy,{\sf StringIso}(\Fx,\Fy,\Gamma,1,W'),W\setminus W')$\;
 }
 compute minimal block system $\FB$ of the action of $\Gamma$ on $W$\;
 compute $\Delta \coloneqq \Gamma_{(\FB)}$\;
 compute transversal $T$ of $\Delta$ in $\Gamma$\;
 \Return $\bigcup_{\delta \in T} {\sf StringIso}(\Fx,\Fy,\Delta,\delta,W)$\;
\end{algorithm}

Overall, Luks's algorithm is formulated in Algorithm \ref{alg:luks}.
The running time of this algorithm heavily depends on the size of the primitive groups involved in the computation.

\paragraph{The Role of Primitive Groups.}

In order to analyse the running time of Luks's algorithm (on a specific input group $\Gamma$) the crucial step is to understand which primitive groups appear along the execution of the algorithm and to bound their size as a function of their degree.
Let us first analyse the situation for GI for graphs of maximum degree $d$.
The crucial observation is that there is a polynomial-time Turing reduction from this problem to the String Isomorphism Problem where the input group $\Gamma$ is contained the class $\mgamma_d$ \cite{Luks82,BabaiL83}.

Let $\Gamma$ be a group.
A \defword{subnormal series} is a sequence of subgroups $\Gamma = \Gamma_0 \trianglerighteq \Gamma_1 \trianglerighteq \dots \trianglerighteq \Gamma_k = \{\id\}$.
The length of the series is $k$ and the groups $\Gamma_{i-1} / \Gamma_{i}$ are the factor groups of the series, $i \in [k]$.
A \defword{composition series} is a strictly decreasing subnormal series of maximal length. 
For every finite group $\Gamma$ all composition series have the same family (considered as a multiset) of factor groups (cf.\ \cite{Rotman99}).
A \defword{composition factor} of a finite group $\Gamma$ is a factor group of a composition series of $\Gamma$.

\begin{Def}
 For $d \geq 2$ let $\mgamma_d$ denote the class of all groups $\Gamma$ for which every composition factor of $\Gamma$ is isomorphic to a subgroup of $S_d$.
\end{Def}

We want to stress the fact that there are two similar classes of groups that have been used in the literature both typically denoted by~$\Gamma_d$.
One of these is the class introduced by Luks~\cite{Luks82} that we denote by $\mgamma_d$, while the other one used in~\cite{BabaiCP82} in particular allows composition factors that are simple groups of Lie type of bounded dimension.

\begin{lemma}[Luks \cite{Luks82}]
 \label{la:gamma-d-closure}
 Let $\Gamma \in \mgamma_d$. Then
 \begin{enumerate}
  \item $\Delta \in \mgamma_d$ for every subgroup $\Delta \leq \Gamma$, and
  \item $\Gamma^{\varphi} \in \mgamma_d$ for every homomorphism $\varphi\colon \Gamma \rightarrow \Delta$.
 \end{enumerate}
\end{lemma}

Now, the crucial advantage of $\mgamma_d$-groups with respect to Luks's algorithm is that the size of primitive $\mgamma_d$-groups is bounded polynomially in the degree for every fixed $d$.

\begin{theorem}
 \label{thm:size-primitive-gamma-d}
 There is a function $f\colon \NN \rightarrow \NN$ such that for every primitive permutation group $\Gamma \in \mgamma_d$ it holds that $|\Gamma| \leq n^{f(d)}$.
 Moreover, the function $f$ can be chosen such that $f(d) = \CO(d)$.
\end{theorem}

The first part of this theorem on the existence of the function $f$
was first proved by Babai, Cameron and Pálfy \cite{BabaiCP82} implying
that Luks's algorithm runs in polynomial time for every fixed number
$d \in \NN$.  Later, it was observed that the function $f$ can
actually be chosen to be linear in $d$ (see, e.g., \cite{LiebeckS99}).

Theorem \ref{thm:size-primitive-gamma-d} allows the analysis of Luks's algorithm when the input
group comes from the class $\mgamma_d$.  Note that the class $\mgamma_d$ is
closed under subgroups and homomorphic images by Lemma
\ref{la:gamma-d-closure} which implies that all groups encountered
during Luks's algorithm are from the class $\mgamma_d$ in case the
input group is a $\mgamma_d$-group.

\begin{cor}
 \label{cor:luks-alg-running-time}
 Luks's Algorithm (Algorithm \ref{alg:luks}) for groups
 $\Gamma \in \mgamma_d$ runs in time $n^{\CO(d)}$.
\end{cor}

\begin{proofsketch}
  In order to analyze the running time of Algorithm \ref{alg:luks} let
  $f(n)$ denote the maximal number of leaves in a recursion tree for a
  $\mgamma_d$-group where $n \coloneqq |W|$ denotes the window size.
  For $n = 1$ it is easy to see that $f(n) = 1$.  Suppose $\Gamma$ is
  not transitive on $W$ and let $n_1 \coloneqq |W'|$ be the size of an
  orbit $W' \subseteq W$.  Then
 \begin{equation}
  \label{eq:recursion-luks-orbits}
  f(n) \leq f(n_1) + f(n - n_1).
 \end{equation}
 Finally, if $\Gamma$ is transitive on $W$, the algorithm computes a minimal block system $\FB$ of $\Gamma[W]$ and performs standard Luks reduction.
 In this case the algorithm performs $t\cdot b$ recursive calls over window size $n/b$ where $t = |\Gamma[\FB]|$.
 Hence,
 \begin{equation}
  \label{eq:recursion-luks-blocks}
  f(n) \leq t \cdot b \cdot f(n/b).
 \end{equation}
 Since $|\Gamma[\FB]|$ is a primitive $\mgamma_d$-group it holds that $t = b^{\CO(d)}$ by Theorem \ref{thm:size-primitive-gamma-d}.
 Combining the bound on $t$ with Equation \eqref{eq:recursion-luks-orbits} and \eqref{eq:recursion-luks-blocks} gives $f(n) = n^{\CO(d)}$.
 Also, each node of the recursion tree only requires computation time polynomial in $n$ and the number of children in the recursion tree (cf.\ Theorem \ref{thm:permutation-group-library}).
 Overall, this gives the desired bound on the running time.
\end{proofsketch} 

In combination with the polynomial-time Turing reduction from the Graph Isomorphism Problem for graphs of maximum degree $d$ to the String Isomorphism Problem for $\mgamma_d$-groups \cite{Luks82,BabaiL83} the former problem can be solved in the same running time $n^{\CO(d)}$.

\medskip

Unfortunately, for the general String Isomorphism Problem, the situation is more complicated.
As a very simple example, symmetric and alternating groups $S_n$ and $A_n$ are primitive groups of exponential size (which means that Luks's algorithm performs a brute-force search over all elements of the group).
However, even in the case of arbitrary input groups, the recursion techniques of Luks are still a powerful tool.
The main reason is that large primitive groups, which build the bottleneck cases for Luks's algorithm, are quite rare and well understood.

Let $m \in \NN$, $t \leq \frac{m}{2}$ and denote $\binom{[m]}{t}$ to be the set of all $t$-element subsets of $[m]$.
Let $S_m^{(t)} \leq \Sym(\binom{[m]}{t})$ denote the natural induced action of $S_m$ on the set $\binom{[m]}{t}$.
Similarly, let $A_m^{(t)} \leq \Sym(\binom{[m]}{t})$ denote the natural induced action of $A_m$ on the set $\binom{[m]}{t}$.
Following Babai \cite{Babai16}, we refer to these groups as the \defword{Johnson groups}.

The classification of large primitive groups by Cameron \cite{Cameron81} states that a primitive permutation group of order $|\Gamma| \geq n^{1+\log n}$ is necessarily a so-called \defword{Cameron group} which involves a large Johnson group.
Here, we only state the following slightly weaker result which is sufficient for most algorithmic purposes.

\begin{theorem}
 \label{thm:luks-obstacles}
 Let $\Gamma \leq \Sym(\Omega)$ be a primitive group of order $|\Gamma| \geq n^{1+\log n}$ 
 where $n$ is greater than some absolute constant.
 Then there is a polynomial-time algorithm computing a normal subgroup $N \unlhd \Gamma$ of index $|\Gamma:N| \leq n$ and an $N$-invariant equipartition $\FB$
 such that $N[\FB]$ is permutationally equivalent to $A_m^{(t)}$ for some $m \geq \log n$.
\end{theorem}

The mathematical part of this theorem follows from \cite{Cameron81,Maroti02} whereas the algorithmic part is resolved in \cite{BabaiLS87} (see also \cite[Theorem 3.2.1]{Babai15}).
The theorem exactly characterizes the obstacle cases of Luks's algorithm and can be seen as the starting point for Babai's algorithm solving the String Isomorphism Problem in quasi-polynomial time \cite{Babai16}.

\subsection{Babai's Algorithm}

Next, we give a brief overview on the main ideas of Babai's quasi-polynomial-time algorithm for the String Isomorphism Problem.
By Theorem \ref{thm:reduction-graph-isomorphism-to-string-isomorphism} this also gives an algorithm for the Graph Isomorphism Problem with essentially the same running time.

As already indicated above the basic strategy of Babai's algorithm is to follow standard Luks recursion until the algorithm encounters an obstacle group which, by Theorem \ref{thm:luks-obstacles}, may be assumed to be a Johnson group.
In order to algorithmically handle the case of Johnson groups Babai’s algorithm utilizes several subroutines based on both group-theoretic techniques and combinatorial approaches like the Weisfeiler-Leman algorithm.

Let us start with discussing the group-theoretic subroutines, specifically the \defword{Local Certificates Routine} which is based on two group-theoretic statements, the \defword{Unaffected Stabilizers Theorem} and the \defword{Affected Orbit Lemma}.

Recall that for a set $M$ we denote by $\Alt(M)$ the alternating group acting with its standard action on the set $M$.
Moreover, following Babai \cite{Babai16}, we refer to the groups $\Alt(M)$ and $\Sym(M)$ as the \defword{giants} where $M$ is an arbitrary finite set.
Let $\Gamma \leq \Sym(\Omega)$.
A \defword{giant representation} is a homomorphism $\varphi\colon \Gamma \rightarrow S_k$ such that $\Gamma^{\varphi} \geq A_k$.
Observe that Johnson groups naturally admit giant representations and thus, the obstacle cases from Theorem \ref{thm:luks-obstacles} also have a giant representation (this is a main feature of an obstacle exploited algorithmically).

Given a string $\Fx\colon \Omega \rightarrow \Sigma$, a group $\Gamma \leq \Sym(\Omega)$ and a giant representation $\varphi\colon \Gamma \rightarrow S_k$, the aim of the Local Certificates Routine is to determine whether $(\Aut_\Gamma(\Fx)^{\varphi}) \geq A_k$ and to compute a meaningful certificate in either case.
To achieve this goal the central tool is to split the set $\Omega$ into \defword{affected} and \defword{non-affected} points.

\begin{Def}[Affected Points, Babai \cite{Babai16}]
 \label{def:affected-points}
 Let $\Gamma \leq \Sym(\Omega)$ be a group and $\varphi\colon\Gamma \rightarrow S_k$ a giant representation.
 Then an element $\alpha \in \Omega$ is \defword{affected by $\varphi$} if $\Gamma_\alpha^{\varphi} \not\geq A_k$.
\end{Def}

We remark that, if $\alpha \in \Omega$ is affected by $\varphi$, then every element in the orbit $\alpha^{\Gamma}$ is affected by $\varphi$.
Hence, the set $\alpha^\Gamma$ is called an \defword{affected orbit} (with respect to $\varphi$).

Let $\Delta \leq \Gamma$ be a subgroup.
We define $\Aff(\Delta,\varphi) \coloneqq \{\alpha \in \Omega \mid \Delta_\alpha^{\varphi} \not\geq A_k\}$.
Observe that, for $\Delta_1 \leq \Delta_2 \leq \Gamma$ it holds that $\Aff(\Delta_1,\varphi) \supseteq \Aff(\Delta_2,\varphi)$.

\begin{algorithm}
 \caption{${\sf LocalCertificates}(\Fx,\Gamma,\varphi)$}
 \label{alg:local-certificates}
 
 \SetKwInOut{Input}{Input}
 \SetKwInOut{Output}{Output}
 \Input{A string $\Fx\colon \Omega \rightarrow \Sigma$, a group $\Gamma \leq \Sym(\Omega)$, and a giant representation $\varphi\colon \Gamma \rightarrow S_k$ such that $k \geq \max\{8,2 + \log_2 n\}$}
 \Output{A non-giant $\Lambda \leq S_k$ with $(\Aut_\Gamma(\Fx))^{\varphi} \leq \Lambda$ or $\Delta \leq \Aut_\Gamma(\Fx)$ with $\Delta^{\varphi} \geq A_k$.}
 \DontPrintSemicolon
 \BlankLine
 $W_0 \coloneqq \emptyset$\;
 $\Gamma_0 \coloneqq \Gamma$\;
 $i \coloneqq 0$\;
 \While{$\Gamma_i^{\varphi} \geq A_k \textup{ \bfseries and } W_i \neq \Aff(\Gamma_i,\varphi)$}{
  $W_{i+1} \coloneqq \Aff(\Gamma_i, \varphi)$\;
  $N \coloneqq \ker(\varphi|_{\Gamma_i})$\;
  $\Gamma_{i+1} \coloneqq \emptyset$\;
  \For{$\gamma \in \Gamma_i^{\varphi}$}{
   compute $\bar \gamma \in \varphi^{-1}(\gamma)$\;
   $\Gamma_{i+1} \coloneqq \Gamma_{i+1} \cup \Aut_{N\bar \gamma}^{W_{i+1}}(\Fx)$
  }
  $i \coloneqq i+1$\;
 }
 \eIf{$\Gamma_i^{\varphi} \not\geq A_k$}{
  \Return $\Gamma_i^{\varphi}$\;
 }{
  \Return $(\Gamma_i)_{(\Omega \setminus W_i)}$
 }
\end{algorithm}

The Local Certificates Routine is described in Algorithm \ref{alg:local-certificates}.
The algorithm computes a sequence of groups $\Gamma = \Gamma_0 \geq \Gamma_1 \geq \dots \geq \Gamma_i$ as well as a sequence of windows $\emptyset = W_0 \subseteq W_1 \subseteq \dots \subseteq W_i$.
Throughout, the algorithm maintains the property that $\Gamma_i = \Aut_\Gamma^{W_i}(\Fx)$, i.e., the algorithm tries to ``approximate'' the automorphism group $\Aut_\Gamma(\Fx)$ taking larger and larger substrings into account.
Observe that $W_{i+1}$ is $\Gamma_i$-invariant and thus, $\Aut_\Gamma(\Fx) \leq \Gamma_{i+1}$.

When growing the window $W_i$ the algorithm adds those points $\alpha \in \Omega$ that are affected (with respect to the current group $\Gamma_i$) and computes the group $\Gamma_{i+1}$ using recursion.
The while-loop is terminated as soon as $\Gamma_i^{\varphi} \not\geq A_k$ or the window $W_i$ stops growing.
In the first case the algorithm returns $\Lambda \coloneqq \Gamma_i^{\varphi}$ which clearly gives the desired outcome.
In the second case the algorithm returns $\Delta \coloneqq (\Gamma_i)_{(\Omega \setminus W_i)}$.
It is easy to verify that $\Delta \leq \Aut_\Gamma(\Fx)$ since all points inside the window $W_i$ are taken into account by the definition of group $\Gamma_i$ and all points outside of $W_i$ are fixed.
The key insight is that $\Delta$ contains a large number of automorphisms, more specifically, $\Delta^{\varphi} \geq A_k$.
This is guaranteed by the \defword{Unaffected Stabilizers Theorem}, one of the main conceptual contributions of Babai's algorithm.

\begin{theorem}[Unaffected Stabilizers Theorem, Babai \cite{Babai16}]
 \label{thm:unaffected-stabilizer-babai}
 Let $\Gamma \leq \Sym(\Omega)$ be a permutation group of degree $n$ and let $\varphi\colon \Gamma \rightarrow S_k$ be a giant representation such that $k > \max\{8,2 + \log n\}$.
 Let $D \subseteq \Omega$ be the set of elements not affected by $\varphi$.
 
 Then $(\Gamma_{(D)})^{\varphi} \geq A_k$. In particular $D \neq \Omega$, that is, at least one point is affected by $\varphi$.
\end{theorem}

Overall, this gives the correctness of the Local Certificates Routine.
So let us analyse the running time.
Here, the crucial property is that the orbits of $N[W]$ are small which guarantees efficient standard Luks reduction when computing $\Gamma_{i+1}$. 

\begin{lemma}[Affected Orbit Lemma, Babai \cite{Babai16}]
 \label{la:kernel-affected-orbits}
 Let $\Gamma \leq \Sym(\Omega)$ be a permutation group and suppose $\varphi\colon \Gamma \rightarrow S_k$ is a giant representation for $k \geq 5$.
 Suppose $A \subseteq \Omega$ is an affected orbit of $\Gamma$ (with respect to $\varphi$). Then every orbit of $\ker(\varphi)$ in $A$ has size at most $|A|/k$.
\end{lemma}

Overall, this means that the Local Certificates Routine runs in time $k!n^{\CO(1)}$ and performs at most $k!n$ many recursive calls to the String Isomorphism Problem over domain size at most $n/k$.
Unfortunately, this not fast enough if $k$ is significantly larger than $\log n$.
The solution to this problem is to not run the Local Certificates Routine for the entire set $[k]$, but only for \defword{test sets} $T \subseteq [k]$ of size $t = \CO(\log n)$.

Let $T \subseteq [k]$ be a test set of size $|T| = t$.
We extend the notion of point- and setwise stabilizers to the image of $\varphi$ and define $\Gamma_{(T)} \coloneqq \varphi^{-1}((\Gamma^\varphi)_{(T)})$ and $\Gamma_{T} \coloneqq \varphi^{-1}((\Gamma^\varphi)_T)$.
The test set $T$ is called \defword{full} if $((\Aut_{\Gamma_T}(\Fx))^\varphi)[T] \geq \Alt(T)$.
A \defword{certificate of fullness} is a group $\Delta \leq \Aut_{\Gamma_T}(\Fx)$ such that $(\Delta^\varphi)[T] \geq \Alt(T)$.
A \defword{certificate of non-fullness} is a non-giant group $\Lambda \leq \Sym(T)$ such that $((\Aut_{\Gamma_T}(\Fx))^\varphi)[T] \leq \Lambda$.
Given a test set $T \subseteq [k]$ we can use the Local Certificates Routine to determine whether $T$ is full and, depending on the outcome, compute a certificate of fullness or non-fullness
(simply run the Local Certificates Routine with input $(\Fx,\Gamma_T,\varphi_T)$ where $\varphi_T(\gamma) = (\gamma^\varphi)[T]$ for all $\gamma \in \Gamma_T$).
Observe that, for $t = \CO(\log n)$, the recursion performed by the Local Certificate Routine for test sets of size $t$ only results in quasi-polynomial running time.

Going back to the original problem of testing isomorphism of strings, the Local Certificates Routine is applied for all test sets of size $t$ as well as both input strings $\Fx,\Fy$.
Based on the results, one can achieve one of the following two outcomes:
\begin{enumerate}
 \item subsets $M_\Fx,M_\Fy \subseteq [k]$ of size at least $\frac{3}{4}k$ and a group $\Delta \leq \Aut_{\Gamma_{M_\Fx}}(\Fx)$ such that $\Delta^\varphi[M_\Fx] \geq \Alt(M_\Fx)$ and $(M_\Fx)^\varphi(\gamma) = M_\Fy$ for all $\gamma \in \Iso_\Gamma(\Fx,\Fy)$, or
 \item two families of $r = k^{\CO(1)}$ many $t$-ary relational structures $(\FA_{\Fx,j})_{j \in [r]}$ and $(\FA_{\Fy,j})_{j \in [r]}$, where each relational structure has domain $[k]$ and only few symmetries, such that
 \[\{\FA_{\Fx,1},\dots,\FA_{\Fx,r}\}^{\varphi(\gamma)} = \{\FA_{\Fy,1},\dots,\FA_{\Fy,r}\}\]
 for all $\gamma \in \Iso_\Gamma(\Fx,\Fy)$.
\end{enumerate}
Here, we shall not formally specify what it means for a structure to have few symmetries, but roughly speaking, this implies that the automorphism group is exponentially smaller than the full symmetric group.

Intuitively speaking, the first option can be achieved if there are many test sets which are full.
In this case, the certificates of fullness (for the string $\Fx$) can be combined to obtain the group $\Delta$.
Otherwise, there are many certificates of non-fullness which can be combined into the relational structures.

Now suppose the first option is satisfied.
For simplicity assume that $M_\Fx = M_\Fy = [k]$ and $\Delta^\varphi = S_k$.
Then $\Iso_\Gamma(\Fx,\Fy) \neq \emptyset$ if and only if $\Iso_{\ker(\varphi)}(\Fx,\Fy) \neq \emptyset$.
Hence, it suffices to recursively determine whether $\Iso_{\ker(\varphi)}(\Fx,\Fy) \neq \emptyset$.
Since $\ker(\varphi)$ is significantly smaller than $\Gamma$ this gives sufficient progress for the recursion to obtain the desired running time.

In the other case the situation is far more involved and here, Babai's algorithm builds on combinatorial methods to achieve further progress.

\medskip

For simplicity, suppose the algorithm is given a single pair of isomorphism-invariant relational structures $\FA \coloneqq \FA_{\Fx,1}$ and $\FB \coloneqq \FA_{\Fy,1}$ (the general case reduces to this case by \defword{individualizing} one of the relational structures).
Recall that $\FA$ and $\FB$ only have few symmetries.
On a high level, the aim of the combinatorial subroutines is to compute ``simpler'' isomorphism-invariant relational structures $\FA^*$ and $\FB^*$ such that $\FA^*$ and $\FB^*$ still have few symmetries and all isomorphisms between $\FA^*$ and $\FB^*$ can be computed in polynomial time.
Roughly speaking, this allows the algorithm to make sufficient progress by computing $\Lambda\lambda \coloneqq \Iso(\FA^*,\FB^*)$ and updating $\Gamma\gamma \coloneqq \varphi^{-1}(\Lambda\lambda)$, i.e., the algorithm continues to compute $\Iso_{\Gamma\gamma}(\Fx,\Fy)$.
Actually, as before, we compute small families of relational structures $\FA^*_1,\dots,\FA^*_m$ and $\FB^*_1,\dots,\FB^*_m$ for some number $m$ which is quasi-polynomial in $k$.
This does not impose any additional problems since, once again, it is possible to individualize a single relational structure.

To achieve the goal, Babai's algorithm builds on two subroutines, the \defword{Design Lemma} and the \defword{Split-or-Johnson Routine}.
The \defword{Design Lemma} first turns the $t$-ary relational structures into (a small family of) graphs (without changing the domain).
For this task, the \defword{Design Lemma} mainly builds on the $t$-dimensional Weisfeiler-Leman algorithm.
The output is passed to the \defword{Split-or-Johnson Routine} which produces one of the following:
\begin{enumerate}
 \item an isomorphism-invariant coloring where each color class $C$ has size $|C| \leq \frac{3}{4}k$, or
 \item an isomorphism-invariant non-trivial equipartition of a subset $M \subseteq [k]$ of size $|M| \geq \frac{3}{4}k$, or
 \item an isomorphism-invariant non-trivial Johnson graph defined on a subset $M \subseteq [k]$ of size $|M| \geq \frac{3}{4}k$.
\end{enumerate}
Here, the Johnson graph $J(m,t)$ is the graph with vertex set $V(J(m,t)) \coloneqq \binom{[m]}{t}$ and edge set $E(J(m,t)) \coloneqq \{XY \mid |X \setminus Y| = 1\}$.
We remark that $\Aut(J(m,t)) = S_m^{(t)}$.

Once again, the \defword{Split-or-Johnson Routine} actually produces a family of such objects.
Note that, for each possible outcome, isomorphisms between objects can be computed efficiently and the number of isomorphisms is exponentially smaller than the entire symmetric group.
This provides the desired progress for the recursive algorithm.

The \defword{Split-or-Johnson Routine} again builds on the Weisfeiler-Leman algorithm as well as further combinatorial tools.
We omit any details here and refer the interested reader to \cite{Babai15,HelfgottBD17}.
Overall, Babai's algorithm results in the following theorem.

\begin{theorem}[Babai \cite{Babai16}]
 The String Isomorphism Problem can be solved in time $n^{\polylog{n}}$.
\end{theorem}

\subsection{Faster Certificates for Groups with Restricted Composition Factors}
\label{subsec:si-gamma-d}

In the following two subsections we discuss extensions of Babai's algorithm for the String Isomorphism Problem.
We start by considering again the String Isomorphism Problem for $\mgamma_d$-groups.
For such an input group, Babai's algorithm never encounters an obstacle case (assuming $d = \CO(\log n)$) and thus, it simply follows Luks's algorithm which solves the problem in time $n^{\CO(d)}$ (see also Corollary \ref{cor:luks-alg-running-time}).
However, in light of the novel group-theoretic subroutines employed in Babai's algorithm, it seems plausible to also hope for improvements in case the input group is a $\mgamma_d$-group.

Looking at Babai's algorithm from a high level, there are two major hurdles that one needs to overcome.
First, we need to analyse the obstacle cases for Luks's algorithm, i.e., one needs to classify primitive $\mgamma_d$-groups of size larger than $n^{\CO(\log d)}$.
Without going into any details, let us just state that by a deep analysis of primitive $\mgamma_d$-groups \cite{GroheNS18}, one can prove that such groups are, once again, essentially Johnson groups in a similar manner as in Theorem \ref{thm:luks-obstacles}.

The second crucial hurdle is the adaptation of the Local Certificates Routine.
Recall that Babai's algorithm applies the Local Certificates Routine to test sets $T$ of size $|T| = t = \CO(\log n)$.
The routine runs in time $t!n^{\CO(1)}$ and performs $t!n$ many recursive calls to the String Isomorphism Problem over domain size at most $n/t$.
The central issue is that, in order to prove correctness building on the Unaffected Stabilizers Theorem, the Local Certificates Routine requires that $t > \max\{8,\log n + 2\}$.
However, to apply the Local Certificates Routine in the setting of $\mgamma_d$-groups, we wish to choose $t = \CO(\log d)$ to attain the desired running time.
Hence, a main obstacle for a faster algorithm for $\mgamma_d$-groups is to obtain a suitable variant of the Unaffected Stabilizers Theorem.
Unfortunately, it is not difficult to see that the natural variant of the Unaffected Stabilizers Theorem (where we restrict the input group to be in $\mgamma_d$ and update the bound on $t$ as desired) does not hold (see, e.g., \cite[Example 5.5.10]{Neuen19}).

The main idea to circumvent this problem \cite{GroheNS18} is to first normalize the input to ensure suitable restrictions on the input group $\Gamma$ that enable us to prove the desired variant of the Unaffected Stabilizers Theorem.
In the following, we only describe the main idea for the normalization.
Given a normalized input, one can provide a suitable variant of the Local Certificates Routine and, building on a more precise analysis of the recursion, extend Babai's algorithm to the setting of $\mgamma_d$-groups.
This results in an algorithm solving the String Isomorphism Problem for $\mgamma_d$-groups in time $n^{\polylog{d}}$.
We remark that the algorithm does not exploit the combinatorial subroutines directly, but one can simply rely on using Babai's algorithm as a black box instead.
In particular, this also simplifies the analysis of the entire algorithm.

To describe the normalization procedure, we follow the framework developed in the second author's PhD thesis \cite{Neuen19}.
A \defword{rooted tree} is a pair $(T,v_0)$ where $T$ is
a directed tree and $v_0 \in V(T)$ is the root of $T$ (all edges are
directed away from the root).  Let $L(T)$ denote the set of leaves of
$T$, i.e., vertices $v \in V(T)$ without outgoing edges.  For
$v \in V(T)$ we denote by $T^{v}$ the subtree of $T$ rooted at vertex $v$.

Let $\Gamma \leq \Sym(\Omega)$ be a permutation group.
A \defword{structure tree} for $\Gamma$ is a rooted tree $(T,v_0)$ such that $L(T) = \Omega$ and $\Gamma \leq (\Aut(T))[\Omega]$, i.e., all $\gamma \in \Gamma$ preserve the structure of $T$.

\begin{lemma}
 Let $\Gamma \leq \Sym(\Omega)$ be a transitive group and $(T,v_0)$ a structure tree for $\Gamma$.
 For every $v \in V(T)$ the set $L(T^{v})$ is a block of $\Gamma$.
 Moreover, $\{L(T^{w}) \mid w \in v^{\Aut(T)}\}$ forms a block system of the group $\Gamma$.
\end{lemma}

Let $\Gamma \leq \Sym(\Omega)$ be a transitive group.
The last lemma implies that every structure tree $(T,v_0)$ gives a sequence of $\Gamma$-invariant partitions $\{\Omega\} = \FB_0 \succ \dots \succ \FB_k = \{\{\alpha\} \mid \alpha \in \Omega\}$ (for each level $h \geq 0$, the blocks associated with the vertices at distance $h$ from the root form one of the partitions).
On the other hand, every such sequence of partitions gives a structure tree $(T,v_0)$ with
\[V(T) = \Omega \cup \bigcup_{i=0,\dots,k-1} \FB_i\]
and
\[E(T) = \{(B,B') \mid B \in \FB_{i-1}, B' \in \FB_i, B' \subseteq B\} \cup \{(B,\alpha) \mid B \in \FB_{k-1}, \alpha \in B\}.\]
The root is $v_0 = \Omega$.

The following definition describes the desired structure of a normalized group.
For $\FB$ a partition of $\Omega$ and $S \subseteq \Omega$ we denote by $\FB[S] \coloneqq \{B \cap S \mid B \in \FB, B \cap S \neq \emptyset\}$ the induced partition on $S$.

\begin{Def}[Almost $d$-ary Sequences of Partitions]
 Let $\Gamma \leq \Sym(\Omega)$ be a group and let $\{\Omega\} = \FB_0 \succ \dots \succ \FB_k = \{\{\alpha\} \mid \alpha \in \Omega\}$ be a sequence of $\Gamma$-invariant partitions.
 The sequence $\FB_0 \succ \dots \succ \FB_k$ is \defword{almost $d$-ary} if for every $i \in [k]$ and $B \in \FB_{i-1}$ it holds that
 \begin{enumerate}
  \item $|\FB_i[B]| \leq d$, or
  \item $\Gamma_B[\FB_i[B]]$ is semi-regular.
 \end{enumerate}
 If the first option is always satisfied the sequence $\FB_0 \succ \dots \succ \FB_k$ is called \defword{$d$-ary}.
 Similarly, a structure tree $(T,v_0)$ for $\Gamma$ is \defword{(almost) $d$-ary} if the corresponding sequence of partitions is.
\end{Def}

Let $\Gamma$ be a group for which there is an almost $d$-ary structure tree $(T,v_0)$.
Now, we can execute Luks's algorithm along the given structure tree $(T,v_0)$ (by always picking the next partition when performing standard Luks reduction).
Intuitively speaking, this allows us to restrict the primitive groups that are encountered during Luks's algorithm.
For a $d$-ary structure tree all primitive groups are subgroups of $S_d$.
For an almost $d$-ary structure tree we additionally need to be able to handle semi-regular groups.
However, handling such groups is simple since they have size at most $n$ where $n$ denotes the size of the permutation domain.

The goal of the \defword{Normalization Routine} is, given an instance $(\Gamma,\Fx,\Fy)$ of the String Isomorphism Problem where $\Gamma \in \mgamma_d$, to compute a normalized equivalent instance $(\Gamma^*,(T,v_0),\Fx^*,\Fy^*)$ such that $(T,v_0)$ forms an almost $d$-ary structure tree for $\Gamma^*$. 
The main tool to achieve the normalization are tree unfoldings of certain \defword{structure graphs}. 

Let $(G,v_0)$ be a rooted acyclic directed graph.
For $v \in V(G)$ define $N^{+}(v) \coloneqq \{w \in V(G) \mid (v,w) \in E(G)\}$ to be the set of outgoing neighbors of $v$.
The \defword{forward degree} of $v$ is $\deg^{+}(v) \coloneqq |N^{+}(v)|$.
A vertex is a \defword{leaf} of $G$ if it has no outgoing neighbors, i.e., $\deg^{+}(v) = 0$.
Let $L(G) = \{v \in V(G) \mid \deg^{+}(v) = 0\}$ denote the set of leaves of $G$.

Let $\Gamma \leq \Sym(\Omega)$ be a permutation group.
A \defword{structure graph} for $\Gamma$ is a triple $(G,v_0,\varphi)$ where $(G,v_0)$ is a rooted acyclic directed graph such that $L(G) = \Omega$ and $\Gamma \leq (\Aut(G))[\Omega]$
and $\varphi\colon \Gamma \rightarrow \Aut(G)$ is a homomorphism such that $(\gamma^{\varphi})[\Omega] = \gamma$ for all $\gamma \in \Gamma$.

Note that each structure tree can be viewed as a structure graph (for trees the homomorphism $\varphi$ is uniquely defined and can be easily computed).
As indicated above the strategy to normalize the action is to consider the tree unfolding of a suitable structure graph.
The permutation domain of the normalized action then corresponds to the leaves of the tree unfolding for which there is a natural action of the group $\Gamma$.

Let $(G,v_0)$ be a rooted acyclic directed graph.
A \defword{branch} of $(G,v_0)$ is a sequence $(v_0,v_1,\dots,v_k)$ such that $(v_{i-1},v_i) \in E(G)$ for all $i \in [k]$.
A branch $(v_0,v_1,\dots,v_k)$ is \defword{maximal} if it cannot be extended to a longer branch, i.e., if $v_k$ is a leaf of $(G,v_0)$.
Let $\Br(G,v_0)$ denote the set of branches of $(G,v_0)$ and $\Br^{*}(G,v_0)$ denote the set of maximal branches.
Note that $\Br^{*}(G,v_0) \subseteq \Br(G,v_0)$.
Also, for a maximal branch $\bar v = (v_0,v_1,\dots,v_k)$ let $L(\bar v) \coloneqq v_k$.
Note that $L(\bar v) \in L(G)$.

For a rooted acyclic directed graph $(G,v_0)$ the \defword{tree unfolding} of $(G,v_0)$ is defined to be the rooted tree $\Unf(G,v_0)$
with vertex set $\Br(G,v_0)$ and edge set
\[E(\Unf(G,v_0)) = \{((v_0,\dots,v_k),(v_0,\dots,v_k,v_{k+1})) \mid (v_0,\dots,v_{k+1}) \in \Br(G,v_0)\}.\]
Note that $L(\Unf(G,v_0)) = \Br^{*}(G,v_0)$, i.e., the leaves of the tree unfolding of $(G,v_0)$ are exactly the maximal branches of $(G,v_0)$.

\begin{example}
 \label{exa:structure-graph-for-johnson-groups}
 Let $m \leq d$, $t \leq \frac{m}{2}$ and consider the Johnson group $\Gamma = A_m^{(t)}$.
 Then a structure graph $(G,v_0,\varphi)$ for $\Gamma \leq \Sym(\binom{[m]}{t})$ can be constructed as follows.
 The vertices of the graph are all subsets of $[m]$ of size at most $t$, i.e.,
 \[V(G) \coloneqq \binom{[m]}{\leq t} = \{X \subseteq [m] \mid |X| \leq t\}.\]
 Two vertices $X$ and $Y$ are connected by an edge if $Y$ is the extension of $X$ by a single element, i.e.,
 \[E(G) \coloneqq \{(X,Y) \mid X \subseteq Y \wedge |Y \setminus X| = 1\}.\]
 The root of the structure graph is $v_0 \coloneqq \emptyset$.
 Note that $L(G) = \binom{[m]}{t}$ as desired.
 Intuitively speaking, $(G,v_0)$ corresponds to the first $t+1$ levels of the subset lattice of the set $[m]$.
 An example is given in Figure \ref{fig:structure-graph-johnson}.
 
 \begin{figure}
  \centering
  \begin{tikzpicture}
   \node[emptyvertex,label={above:$\emptyset$}] (r) at (5.4,3.2) {};
   \foreach \i in {1,...,5}{
    \node[emptyvertex,label={right:$\{\i\}$}] (u\i) at (2.4*\i - 1.8,1.6) {};
    \draw[thick,->] (r) edge (u\i);
   }
   \foreach \a/\b [count=\i] in {1/2,1/3,2/3,1/4,2/4,1/5,2/5,3/4,3/5,4/5}{
    \node[emptyvertex,label={below:$\{\a,\b\}$}] (l\i) at (1.2*\i - 1.2,0) {};
    \draw[thick,->] (u\a) edge (l\i);
    \draw[thick,->] (u\b) edge (l\i);
   }
  \end{tikzpicture}
  \caption{A structure graph for the Johnson group $A_5^{(2)}$}
  \label{fig:structure-graph-johnson}
 \end{figure}
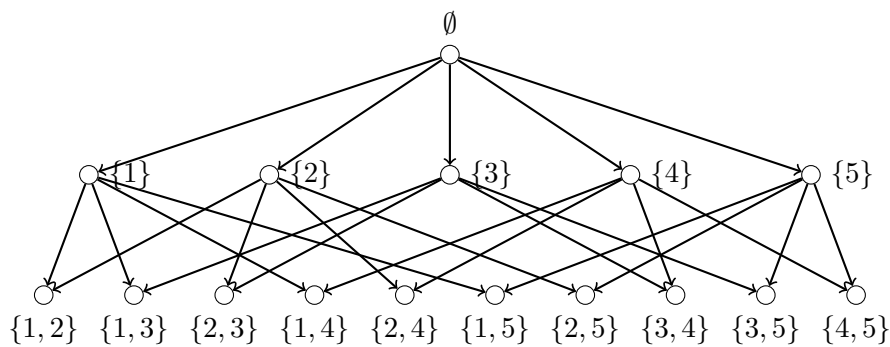
 
 For $\gamma \in A_m$ let $\gamma^{(t)}$ be the element obtained from the natural action of $\gamma$ on $\binom{[m]}{t}$.
 So $\Gamma = \{\gamma^{(t)} \mid \gamma \in A_m\}$.
 Let $\varphi\colon\Gamma\rightarrow\Aut(G)$ be defined by $X^{\varphi(\gamma^{(t)})} = X^{\gamma}$ where $X \in V(G)$.
 Then $(G,v_0,\varphi)$ is a structure graph for $\Gamma$.
 
 Now consider the tree unfolding $\Unf(G,v_0)$ of the graph $(G,v_0)$.
 A maximal branch $\bar v \in \Br^{*}(G,v_0)$ is a sequence $\bar v = (X_0,X_1,\dots,X_t)$ where $X_{i-1} \subseteq X_i$ and $|X_i| = i$ for all $i \in [t]$.
 This gives an ordering of the elements in $X_t$.
 Indeed, it is not difficult to see that there is a one-to-one correspondence between the maximal branches $\Br^{*}(G,v_0)$ and the set $[m]^{\langle t\rangle}$ of ordered $t$-tuples over the set $[m]$ with pairwise distinct elements.
 
 The group $A_m$ also acts naturally on the set $[m]^{\langle t\rangle}$.
 Let $A_m^{\langle t\rangle}$ be the permutation group obtained from this action.
 Now it can be observed that $\Unf(G,v_0)$ gives a structure tree for $A_m^{\langle t\rangle}$.
 Also, the degree of the permutation group $A_m^{\langle t\rangle}$ is only slightly larger than the degree of $A_m^{(t)}$.
 Indeed,
 \begin{equation}
  \label{eq:structure-graph-johnson-size}
  \left|\binom{[m]}{t}\right|^{\log m} = \binom{m}{t}^{\log m} \geq \left(\frac{m}{t}\right)^{t \log m} \geq 2^{t \log m} = m^{t} \geq |[m]^{\langle t\rangle}|.
 \end{equation}
\end{example}

In order to generalize this example to obtain a normalization for all groups one can first observe that a group $\Gamma$ naturally acts on the set of maximal branches of a structure graph and additionally, the tree unfolding of the structure graph forms a structure tree for this action.

\begin{lemma}
 \label{la:standard-action-on-branches}
 Let $\Gamma \leq \Sym(\Omega)$ be a permutation group and let $(G,v_0,\varphi)$ be a structure graph for $\Gamma$.
 Then there is an action $\psi\colon \Gamma \rightarrow \Sym(\Br^{*}(G))$ on the set of maximal branches of $(G,v_0)$ such that
 \begin{enumerate}
  \item $L(\bar v^{\psi(\gamma)}) = (L(\bar v))^{\gamma}$ for all $\bar v \in \Br^{*}(G)$ and $\gamma \in \Gamma$, and
  \item $\Unf(G,v_0)$ forms a structure tree for $\Gamma^{\psi}$.
 \end{enumerate}
 Moreover, given the group $\Gamma$ and the structure graph $(G,v_0,\varphi)$, the homomorphism $\psi$ can be computed in time polynomial in $|\Br^{*}(G,v_0)|$.
\end{lemma}

Hence, for the normalization, the strategy is to first compute a suitable structure graph $(G,v_0,\varphi)$ for $\Gamma$.
Then we update the group by setting $\Gamma^* \coloneqq \Gamma^\psi$ to be the \defword{standard action on the set of maximal branches} described in Lemma \ref{la:standard-action-on-branches}.
Here, the tree unfolding $\Unf(G,v_0)$ should form an almost $d$-ary structure tree for $\Gamma^*$.
Towards this end, we say that $(G,v_0,\varphi)$ is an \defword{almost $d$-ary} structure graph for $\Gamma$ if $\Unf(G,v_0)$ forms an almost $d$-ary structure tree for $\Gamma^*$.
Moreover, we also update the strings $\Fx,\Fy\colon \Omega \rightarrow \Sigma$ by setting $\Fx^*\colon \Br^{*}(G)\rightarrow\Sigma\colon \bar v \mapsto \Fx(L(\bar v))$ and similarly $\Fy^*\colon \Br^{*}(G)\rightarrow\Sigma\colon \bar v \mapsto \Fy(L(\bar v))$.
One can easily show that
\[(\Iso_\Gamma(\Fx,\Fy))^\psi = \Iso_{\Gamma^*}(\Fx^*,\Fy^*).\]
Hence, to complete the normalization, it only remains to find an almost $d$-ary structure graph for which the number of maximal branches is not too large.

\begin{theorem}[\cite{GroheNS18,Neuen19}]
 \label{thm:construct-structure-graph}
 Let $\Gamma \leq \Sym(\Omega)$ be a $\mgamma_d$-group.
 Then there is an almost $d$-ary structure graph $(G,v_0,\varphi)$ for $\Gamma$ such that $|\Br(G,v_0)| \leq n^{\CO(\log d)}$.
 Moreover, there is an algorithm computing such a structure graph in time polynomial in the size of $G$.
\end{theorem}

The proof is the theorem is far from trivial and builds on the characterization of large primitive $\mgamma_d$-groups (see \cite{GroheNS18}).

\begin{cor}
 \label{cor:action-normalization-gamma-d}
 There is a Turing-reduction from the String Isomorphism Problem for $\mgamma_d$-groups to
 the String Isomorphism Problem for groups equipped with an almost $d$-ary structure tree running in time $n^{\CO(\log d)}$.
\end{cor}

As already indicated above, having normalized the input to the String Isomorphism Problem, one can prove an adaptation of the Unaffected Stabilizers Theorem and extend Babai's algorithm to the setting of $\mgamma_d$-groups.

\begin{theorem}[Grohe, Neuen and Schweitzer~\cite{GroheNS18}]
 The String Isomorphism Problem for $\mgamma_d$-groups can be solved in time $n^{\polylog{d}}$.
\end{theorem}

Recall that there is a polynomial-time Turing reduction from the Graph Isomorphism Problem for graphs of maximum degree $d$ to the String Isomorphism Problem for $\mgamma_d$-groups \cite{Luks82,BabaiL83}.
Hence, the Graph Isomorphism Problem for graphs of maximum degree $d$ can be solved in time $n^{\polylog{d}}$.
Actually, we reprove this result in Section \ref{sec:parameterized-isomorphism-tests} using slightly different tools.

\subsection{From Strings to Hypergraphs}
\label{subsec:sets-of-strings}

Next, we present another extension of the presented methods taking more general input structures into account.
Specifically, up to this point, we only considered strings for the input structures which are quite restrictive.
In this subsection, we briefly describe how to extend the methods to isomorphism testing for hypergraphs.

A hypergraph is a pair $\CH = (V,\CE)$ where $V$ is a finite vertex set and $\CE \subseteq 2^{V}$ (where $2^{V}$ denotes the power set of $V$).
As for strings, we are interested in the isomorphism problem between hypergraphs where we are additionally given a permutation group that restricts possible isomorphisms between the two given hypergraphs.
More precisely, in this work, the \defword{Hypergraph Isomorphism Problem} takes as input two hypergraphs $\CH_1 = (V,\CE_1)$ and $\CH_2 = (V,\CE_2)$ over the same vertex set and a permutation group $\Gamma \leq \Sym(V)$ (given by a set of generators),
and asks whether is some $\gamma \in \Gamma$ such that $\gamma\colon \CH_1 \cong \CH_2$ (i.e., $\gamma$ is an isomorphism from $\CH_1$ to $\CH_2$).

For the purpose of designing an algorithm it is actually more convenient to consider the following equivalent problem.
The \defword{Set-of-Strings Isomorphism Problem} takes as input two sets $\FX = \{\Fx_1,\dots,\Fx_m\}$ and $\FY = \{\Fy_1,\dots,\Fy_m\}$ where $\Fx_i,\Fy_i \colon \Omega \rightarrow \Sigma$ are strings, and a group $\Gamma \leq \Sym(\Omega)$,
and asks whether there is some $\gamma \in \Gamma$ such that $\FX^{\gamma} \coloneqq \{\Fx_1^{\gamma},\dots,\Fx_m^{\gamma}\} = \FY$.

\begin{theorem}
 \label{thm:eq-sets-of-strings-hypergraphs}
 The Hypergraph Isomorphism Problem for $\mgamma_d$-groups is
 polynomial-time equivalent to the Set-of-Strings Isomorphism Problem
 for $\mgamma_d$-groups under many-one reductions.
\end{theorem}

\begin{proofsketch}
 A hypergraph $\CH = (V,\CE)$ can be translated into a set $\FX \coloneqq \{\Fx_E \mid E \in \CE\}$ over domain $\Omega \coloneqq V$ where $\Fx_E$ is the characteristic function of $E$, i.e., $\Fx(v) = 1$ if $v \in E$ and $\Fx(v) = 0$ if $v \notin E$.
 
 In the other direction, a set of strings $\FX$ over domain $\Omega$ and alphabet $\Sigma$ can be translated into a hypergraph $\CH = (V,\CE)$ where $V \coloneqq \Omega \times \Sigma$ and $\CE \coloneqq \{E_\Fx \mid \Fx \in \FX\}$
 where $E_\Fx \coloneqq \{(\alpha,\Fx(\alpha)) \mid \alpha \in \Omega)\}$.
 The group $\Gamma \leq \Sym(\Omega)$ is translated into the natural action of $\Gamma$ on $V$ defined via $(\alpha,a)^\gamma \coloneqq (\alpha^\gamma,a)$.
 
 It is easy to verify that both translations preserve isomorphisms.
\end{proofsketch}

For the purpose of building a recursive algorithm, we consider a slightly different problem that crucially allows us to modify instances in a certain way.

Let $\Gamma \leq \Sym(\Omega)$ be a group and let $\FP$ be a $\Gamma$-invariant partition of the set $\Omega$.
A \defword{$\FP$-string} is a pair $(P,\Fx)$ where $P \in \FP$ and $\Fx\colon P \rightarrow \Sigma$ is a string over a finite alphabet $\Sigma$.
For $\sigma \in \Sym(\Omega)$ the string $\Fx^{\sigma}$ is defined by $\Fx^{\sigma}\colon P^{\sigma} \rightarrow \Sigma\colon \alpha \mapsto \Fx(\alpha^{\sigma^{-1}})$.
A permutation $\sigma \in \Sym(\Omega)$ is a \defword{$\Gamma$-isomorphism} from $(P,\Fx)$ to a second $\FP$-string $(Q,\Fy)$ if $\sigma \in \Gamma$ and $(P^{\sigma},\Fx^{\sigma}) = (Q,\Fy)$.

The \defword{Generalized String Isomorphism Problem} takes as input a permutation group $\Gamma \leq \Sym(\Omega)$,
a $\Gamma$-invariant partition $\FP$ of $\Omega$,
and $\FP$-strings $(P_1,\Fx_1),\dots,(P_m,\Fx_m)$ and $(Q_1,\Fy_1),\dots,(Q_m,\Fy_m)$,
and asks whether there is some $\gamma \in \Gamma$ such that
\begin{equation*}
\{(P_1^{\gamma},\Fx_1^{\gamma}),\dots,(P_m^{\gamma},\Fx_m^{\gamma})\} = \{(Q_1,\Fy_1),\dots,(Q_m,\Fy_m)\}.
\end{equation*}

We denote $\FX = \{(P_1,\Fx_1),\dots,(P_m,\Fx_m)\}$ and $\FY = \{(Q_1,\Fy_1),\dots,(Q_m,\Fy_m)\}$.
Additionally, $\Iso_\Gamma(\FX,\FY)$ denotes the set of $\Gamma$-isomorphisms from $\FX$ to $\FY$ and $\Aut_\Gamma(\FX) \coloneqq \Iso_\Gamma(\FX,\FX)$.

It is easy to see that the Set-of-Strings Isomorphism Problem forms a special case of the Generalized String Isomorphism Problem where $\FP$ is the trivial partition consisting of one block.

For the rest of this subsection we denote by $n \coloneqq |\Omega|$ the size of the domain, and $m$ denotes the size of $\FX$ and $\FY$ (we always assume $|\FX| = |\FY|$, otherwise the problem is trivial). 
The goal is to sketch an algorithm that solves the Generalized String Isomorphism Problem for $\mgamma_d$-groups in time $(n+m)^{\polylog{d}}$.

As a starting point it was already observed by Miller \cite{Miller83b} that Luks's algorithm can be easily extended to hypergraphs resulting in an isomorphism test running in time $(n+m)^{\CO(d)}$ for $\mgamma_d$-groups.
Similar to the previous subsection, the main obstacle we are facing to obtain a more efficient algorithm is the adaptation of the Local Certificates Routine.

Let $\Gamma \leq \Sym(\Omega)$ be a $\mgamma_d$-group, $\FP$ a $\Gamma$-invariant partition of $\Omega$, and $\FX,\FY$ two sets of $\FP$-strings.
Recall that, in a nutshell, the Local Certificates Routine considers a $\Gamma$-invariant window $W \subseteq \Omega$ such that $\Gamma[W] \leq \Aut(\FX[W])$ (i.e., the group $\Gamma$ respects $\FX$ restricted to the window $W$) and aims at creating automorphisms of the entire structure $\FX$ (from the local information that $\Gamma$ respects $\FX$ on the window $W$).
In order to create these automorphisms the Local Certificates Routine considers the group $\Gamma_{(\Omega \setminus W)}$ fixing every point outside of $W$.
Remember that the Unaffected Stabilizers Theorem (resp.\ the variant
suitable for $\mgamma_d$-groups) guarantees that the group
$\Gamma_{(\Omega \setminus W)}$ is large.

For the String Isomorphism Problem it is easy to see that
$\Gamma_{(\Omega \setminus W)}$ consists only of automorphisms of the
input string (assuming $\Gamma$ respects the input string $\Fx$ on the window $W$) since there are no dependencies between the positions
within the window $W$ and outside of $W$.
However, for the Generalized String Isomorphism Problem, this is not true anymore.
As a simple example, suppose the input is a graph on vertex set $\Omega$ (which, in particular, can be interpreted as a hypergraph and translated into a set of strings) and the edges between $W$ and $\Omega \setminus W$ form a perfect matching.
Then a permutation $\gamma \in \Gamma_{(\Omega \setminus W)}$ is not necessarily an automorphism of $G$ even if it respects $G[W]$, since it may not preserve the edges between $W$ and $\Omega \setminus W$.
Actually, since the edges between $W$ and $\Omega \setminus W$ form a perfect matching, the only automorphism of $G$ in the group $\Gamma_{(\Omega \setminus W)}$ is the identity mapping.

In other words, in order to compute $\Aut(\FX)$, it is not possible to consider $\FX[W]$ and $\FX[\Omega \setminus W]$ independently as is done by the Local Certificates Routine.

The solution to this problem is guided by the following simple observation.
Suppose that $\FX[W]$ is \defword{simple}, i.e.,
\[m_{\FX[W]}(P) \coloneqq |\{\Fx[W \cap P] \mid (P,\Fx) \in \FX, W \cap P \neq \emptyset\}| = 1\]
for all $P \in \FP[W]$ (in other words, when restricted to $W$, each block $P$ only contains one string).
In this case it is possible to consider $\FX[W]$ and $\FX[\Omega \setminus W]$ independently since there can be no more additional dependencies between $W$ and $\Omega \setminus W$.

Let $W \subseteq \Omega$ be a $\Gamma$-invariant set such that $\Gamma[W] \leq \Aut(\FX[W])$ (this is the case during the Local Certificates Routine).
In order to solve the problem described above in general, the basic idea is to introduce another normalization procedure, referred to as \defword{Simplification Routine}, modifying the instance in such a way that $\FX[W]$ becomes simple.
Eventually, this allows us to extend the Local Certificates Routine to the setting of the Generalized String Isomorphism Problem.

In the following we briefly describe the \defword{Simplification Routine} which exploits the specific definition of the Generalized String Isomorphism Problem.
Consider a set $P \in \FP$.
In order to ``simplify'' the instance we define an equivalence
relation on the set $\FX[[P]] \coloneqq \{\Fx \mid (P,\Fx) \in \FX\}$ of all strings on $P$.
Two $\FP$-strings $(P,\Fx_1)$ and $(P,\Fx_2)$ are $W$-equivalent if they are identical on the window $W$, i.e., $\Fx_1[W \cap P] = \Fx_2[W \cap P]$.
For each equivalence class we create a new block $P'$ containing exactly the strings from the equivalence class.
Since the group $\Gamma$ respects the induced sub-instance $\FX[W]$ it naturally acts on the equivalence classes.
This process is visualized in Figure \ref{fig:simplify-on-window-ext} and formalized below.

\begin{figure}
  \centering
  \scalebox{0.9}{
  \begin{tikzpicture}
   \node at (-4.2,0.8) {$\FX$};
   
   \draw[fill,gray!60] (-0.15,-0.15) rectangle (1.2,1.8);
   \draw[fill,gray!60] (3.15,-0.15) rectangle (4.2,1.8);
   
   \node at (1.35,2.4) {$P_1$};
   \node at (4.65,2.4) {$P_2$};
   
   \draw[thick] (-0.3,-0.3) -- (6.3,-0.3);
   \draw[thick] (-0.3,1.95) -- (6.3,1.95);
   \draw[thick] (-0.3,-0.3) -- (-0.3,1.95);
   \draw[thick] (3,-0.3) -- (3,1.95);
   \draw[thick] (6.3,-0.3) -- (6.3,1.95);
   
   \draw (0,0) -- (2.7,0);
   \draw (0,0.3) -- (2.7,0.3);
   \draw (0,0.45) -- (2.7,0.45);
   \draw (0,0.75) -- (2.7,0.75);
   \draw (0,0.9) -- (2.7,0.9);
   \draw (0,1.2) -- (2.7,1.2);
   \draw (0,1.35) -- (2.7,1.35);
   \draw (0,1.65) -- (2.7,1.65);
   \foreach \i in {0,...,9}{
    \draw (\i*0.3,0) -- (\i*0.3,0.3);
    \draw (\i*0.3,0.45) -- (\i*0.3,0.75);
    \draw (\i*0.3,0.9) -- (\i*0.3,1.2);
    \draw (\i*0.3,1.35) -- (\i*0.3,1.65);
   }
   \foreach \a [count=\i] in {a,b,a,a,a,b,b,a,b}{
    \node at (-0.15 + \i*0.3,0.15) {\scriptsize $\a$};
   }
   \foreach \a [count=\i] in {a,a,a,b,a,b,a,a,a}{
    \node at (-0.15 + \i*0.3,0.6) {\scriptsize $\a$};
   }
   \foreach \a [count=\i] in {a,a,a,b,a,a,a,b,a}{
    \node at (-0.15 + \i*0.3,1.05) {\scriptsize $\a$};
   }
   \foreach \a [count=\i] in {a,b,a,a,a,b,a,b,a}{
    \node at (-0.15 + \i*0.3,1.5) {\scriptsize $\a$};
   }
   
   \draw (3.3,0) -- (6,0);
   \draw (3.3,0.3) -- (6,0.3);
   \draw (3.3,0.45) -- (6,0.45);
   \draw (3.3,0.75) -- (6,0.75);
   \draw (3.3,0.9) -- (6,0.9);
   \draw (3.3,1.2) -- (6,1.2);
   \draw (3.3,1.35) -- (6,1.35);
   \draw (3.3,1.65) -- (6,1.65);
   \foreach \i in {0,...,9}{
    \draw (3.3+\i*0.3,0) -- (3.3+\i*0.3,0.3);-0.3
    \draw (3.3+\i*0.3,0.45) -- (3.3+\i*0.3,0.75);
    \draw (3.3+\i*0.3,0.9) -- (3.3+\i*0.3,1.2);
    \draw (3.3+\i*0.3,1.35) -- (3.3+\i*0.3,1.65);
   }
   \foreach \a [count=\i] in {a,a,b,b,b,b,b,a,a}{
    \node at (3.15 + \i*0.3,0.15) {\scriptsize $\a$};
   }
   \foreach \a [count=\i] in {a,b,a,a,a,a,a,b,b}{
    \node at (3.15 + \i*0.3,0.6) {\scriptsize $\a$};
   }
   \foreach \a [count=\i] in {a,b,a,a,a,b,a,a,b}{
    \node at (3.15 + \i*0.3,1.05) {\scriptsize $\a$};
   }
   \foreach \a [count=\i] in {a,b,a,a,b,a,a,b,b}{
    \node at (3.15 + \i*0.3,1.5) {\scriptsize $\a$};
   }
   
   \draw[thick,->] (1,-0.4) -- (-1.95,-1.4);
   \draw[thick,->] (2.3,-0.4) -- (1.25,-1.4);
   \draw[thick,->] (3.7,-0.4) -- (4.55,-1.4);
   \draw[thick,->] (5,-0.4) -- (7.85,-1.4);
   
   \node at (-4.2,-2.4) {$\FX'$};
   
   \draw[fill,gray!60] (-3.45,-3.15) rectangle (-2.1,-2.1);
   \draw[fill,gray!60] (-0.15,-3.15) rectangle (1.2,-2.1);
   \draw[fill,gray!60] (3.15,-3.15) rectangle (4.2,-1.65);
   \draw[fill,gray!60] (6.45,-3.15) rectangle (7.5,-2.55);
   
   \node at (-1.95,-3.75) {$P_1 \times \{abaa\}$};
   \node at (1.35,-3.75) {$P_1 \times \{aaab\}$};
   \node at (4.65,-3.75) {$P_2 \times \{aba\}$};
   \node at (7.95,-3.75) {$P_2 \times \{aab\}$};
   
   \draw[thick] (-3.6,-3.3) -- (9.6,-3.3);
   \draw[thick] (-3.6,-1.5) -- (9.6,-1.5);
   \draw[thick] (-3.6,-3.3) -- (-3.6,-1.5);
   \draw[thick] (-0.3,-3.3) -- (-0.3,-1.5);
   \draw[thick] (3,-3.3) -- (3,-1.5);
   \draw[thick] (6.3,-3.3) -- (6.3,-1.5);
   \draw[thick] (9.6,-3.3) -- (9.6,-1.5);
   
   \draw (0,-3) -- (2.7,-3);
   \draw (0,-2.7) -- (2.7,-2.7);
   \draw (0,-2.55) -- (2.7,-2.55);
   \draw (0,-2.25) -- (2.7,-2.25);
   \draw (-3.3,-3) -- (-0.6,-3);
   \draw (-3.3,-2.7) -- (-0.6,-2.7);
   \draw (-3.3,-2.55) -- (-0.6,-2.55);
   \draw (-3.3,-2.25) -- (-0.6,-2.25);
   \foreach \i in {0,...,9}{
    \draw (-3.3+\i*0.3,-3) -- (-3.3+\i*0.3,-2.7);
    \draw (-3.3+\i*0.3,-2.55) -- (-3.3+\i*0.3,-2.25);
    \draw (\i*0.3,-3) -- (\i*0.3,-2.7);
    \draw (\i*0.3,-2.55) -- (\i*0.3,-2.25);
   }
   \foreach \a [count=\i] in {a,b,a,a,a,b,b,a,b}{
    \node at (-3.45 + \i*0.3,-2.85) {\scriptsize $\a$};
   }
   \foreach \a [count=\i] in {a,a,a,b,a,b,a,a,a}{
    \node at (-0.15 + \i*0.3,-2.85) {\scriptsize $\a$};
   }
   \foreach \a [count=\i] in {a,a,a,b,a,a,a,b,a}{
    \node at (-0.15 + \i*0.3,-2.4) {\scriptsize $\a$};
   }
   \foreach \a [count=\i] in {a,b,a,a,a,b,a,b,a}{
    \node at (-3.45 + \i*0.3,-2.4) {\scriptsize $\a$};
   }
   
   \draw (3.3,-3) -- (6,-3);
   \draw (3.3,-2.7) -- (6,-2.7);
   \draw (3.3,-2.55) -- (6,-2.55);
   \draw (3.3,-2.25) -- (6,-2.25);
   \draw (3.3,-2.1) -- (6,-2.1);
   \draw (3.3,-1.8) -- (6,-1.8);
   \draw (6.6,-3) -- (9.3,-3);
   \draw (6.6,-2.7) -- (9.3,-2.7);
   \foreach \i in {0,...,9}{
    \draw (3.3+\i*0.3,-3) -- (3.3+\i*0.3,-2.7);
    \draw (3.3+\i*0.3,-2.55) -- (3.3+\i*0.3,-2.25);
    \draw (3.3+\i*0.3,-2.1) -- (3.3+\i*0.3,-1.8);
    \draw (6.6+\i*0.3,-3) -- (6.6+\i*0.3,-2.7);
   }
   \foreach \a [count=\i] in {a,a,b,b,b,b,b,a,a}{
    \node at (6.45 + \i*0.3,-2.85) {\scriptsize $\a$};
   }
   \foreach \a [count=\i] in {a,b,a,a,a,a,a,b,b}{
    \node at (3.15 + \i*0.3,-2.85) {\scriptsize $\a$};
   }
   \foreach \a [count=\i] in {a,b,a,a,a,b,a,a,b}{
    \node at (3.15 + \i*0.3,-2.4) {\scriptsize $\a$};
   }
   \foreach \a [count=\i] in {a,b,a,a,b,a,a,b,b}{
    \node at (3.15 + \i*0.3,-1.95) {\scriptsize $\a$};
   }
   
  \end{tikzpicture}
  }
  \caption{A set $\FX$ of $\FP$-strings is given in the top and the ``simplified'' instance $\FX'$ is given below. The window $W$ is marked in gray.
   Note that $\FX'[W']$ is simple where $W'$ denotes the window marked in gray in the bottom part of the figure.}
  \label{fig:simplify-on-window-ext}
\end{figure}
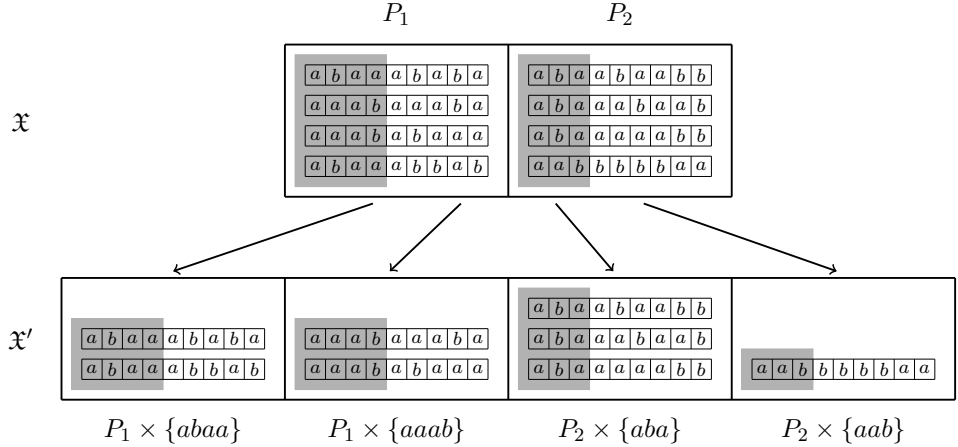

Let $\Omega' \coloneqq \bigcup_{P \in \FP} P \times \{\Fx[W \cap P] \mid (P,\Fx) \in \FX\}$
and $\FP' \coloneqq \{P \times \{\Fx[W \cap P]\} \mid (P,\Fx) \in \FX\}.$
Also define
\[\FX' \coloneqq \left\{\Bigl(P \times \{\Fx[W \cap P]\},\Fx'\Bigr) \;\;\Big|\;\; (P,\Fx) \in \FX\right\}\]
where $\Fx'\colon P \times \{\Fx[W \cap P]\} \rightarrow \Sigma\colon (\alpha,\Fx[W \cap P]) \mapsto \Fx(\alpha)$.
The set $\FY'$ is defined similarly for the instance $\FY$.
Note that $\FX'$ and $\FY'$ are sets of $\FP'$-strings.
Finally, the group $\Gamma$ faithfully acts on the set $\Omega'$ via $(\alpha,\Fz)^{\gamma} = (\alpha^{\gamma},\Fz^{\gamma})$ yielding an injective homomorphism $\psi \colon \Gamma \rightarrow \Sym(\Omega')$.
Define $\Gamma' \coloneqq \Gamma^{\psi}$.
It can be easily checked that the updated instance is equivalent to the original instance.
Also, $\FX'[W']$ is simple where $W' \coloneqq \{(\alpha,\Fz) \in \Omega' \mid \alpha \in W\}$.

While this simplification allows us to treat $\FX'[W']$ and $\FX'[W' \setminus \Omega']$ independently and thus solves the above problem, it creates several additional issues that need to be addressed.
First, this modification may destroy the normalization property (i.e., the existence of an almost $d$-ary sequence of partitions).
As a result, the Local Certificates Routine constantly needs to re-normalize the input instances which requires a precise analysis of the increase in size occurring from the re-normalization.
In turn, the re-normalization of instances creates another problem.
In the \defword{Aggregation of Local Certificates} (where the local certificates are combined into a relational structure), the outputs of the Local Certificates Routine are compared with each other requiring the outputs to be isomorphism-invariant.
However, the re-normalization procedure is not isomorphism-invariant.
The solution to this problem is to run the Local Certificates Routine in parallel on all pairs of test sets compared later on.
This way, one can ensure that all instances are normalized in the same way.

Overall, these ideas allow us to obtain the following theorem.

\begin{theorem}[Neuen \cite{Neuen20}]
 \label{thm:sets-of-strings}
 The Generalized String Isomorphism Problem for $\mgamma_d$-groups can be solved in time $(n+m)^{\polylog{d}}$ where $n$ denotes the size of the domain and $m$ the number of strings in the input sets.
\end{theorem}

Note that this gives an algorithm for the Set-of-Strings Isomorphism Problem as well as the Hypergraph Isomorphism Problem for $\mgamma_d$-groups with the same running time.
In particular, the Hypergraph Isomorphism Problem without any input
group can be solved in time $(n+m)^{\polylog{n}}$.
This is the fasted known algorithm for the problem. A different
algorithm with the same running time was obtained by Wiebking \cite{Wiebking20}.

\section{Quasi-Polynomial Parameterized Algorithms for Isomorphism Testing}
\label{sec:parameterized-isomorphism-tests}

In this section we present several applications of the results presented above for isomorphism testing on restricted classes of graphs.
Towards this end, we first introduce the notion of \defword{$t$-CR-bounded} graphs recently introduced in \cite{Neuen20} and build an isomorphism test for such graphs based on Theorem \ref{thm:sets-of-strings}.
It turns out that the notion of $t$-CR-bounded graphs forms a powerful tool when it comes to the task of designing isomorphism tests for restricted classes of graphs.
In this direction we build a series of reductions for the isomorphism problem for well-known parameterized classes of graphs to the isomorphism problem for $t$-CR-bounded graphs leading to the most efficient algorithms for isomorphism testing for mentioned classes.
An overview on the reductions can be found in Figure \ref{fig:reductions}.

\begin{figure}
 \centering
 \begin{tikzpicture}[every edge/.style={draw,thick,->}]
  
  \node[draw, thick,text width = 3cm,text centered, rounded corners] (HI) at (6.4,7.2) {\footnotesize Set-of-Strings Isomorphism for $\mgamma_d$-groups};
  \node[draw, thick,text width = 3cm,text centered, rounded corners] (SI) at (1,6) {\footnotesize String Isomorphism for $\mgamma_d$-groups};
  \node[draw, thick,text width = 3cm,text centered, rounded corners] (CR) at (6.4,4.8) {\footnotesize GI for $t$-CR-bounded graphs};
  \node[draw, thick,text width = 3cm,text centered, rounded corners] (TS) at (4.2,2.4) {\footnotesize GI for graphs excluding $K_{3,h}$ as topological subgraph};
  \node[draw, thick,text width = 3cm,text centered, rounded corners] (HA) at (8.8,2.4) {\footnotesize GI for graphs of Hadwiger number $h$};
  \node[draw, thick,text width = 3cm,text centered, rounded corners] (DE) at (2,0) {\footnotesize GI for graphs of maximum degree $d$};
  \node[draw, thick,text width = 3cm,text centered, rounded corners] (GE) at (6.4,0) {\footnotesize GI for graphs of genus $g$};
  \node[draw, thick,text width = 3cm,text centered, rounded corners] (TW) at (10.8,0) {\footnotesize GI for graphs of tree width $k$};
  
  \draw (SI) edge (HI);
  \draw (CR) edge node[right] {\footnotesize $d \coloneqq t$} (HI);
  \draw (TS) edge node[left] {\footnotesize $t \coloneqq 7h-1$} (CR);
  \draw (HA) edge[dashed] node[right] {\footnotesize $t \coloneqq \CO(h^4)$} (CR);
  \draw ([xshift=-1cm]DE.north) edge (SI);
  \draw (DE) edge node[right] {\footnotesize $h \coloneqq d+1$} (TS);
  \draw (GE) edge node[right] {\footnotesize $h \coloneqq 4g+3$} (TS);
  \draw (GE) edge node[right] {\footnotesize $h \coloneqq g+6$} (HA);
  \draw (TW) edge node[right] {\footnotesize $h \coloneqq k+1$} (HA);
  \draw ([xshift=0.6cm]TW.north) edge[dashed, bend right=45] node[right] {\footnotesize $t \coloneqq k$} (CR.east);
 \end{tikzpicture}
 \caption{Dependencies between isomorphism problems considered in this paper. Each solid edge represents a polynomial-time Turing reduction where the parameter is set according to the edge label. The dashed edges represent fpt reductions.}
 \label{fig:reductions}
\end{figure}
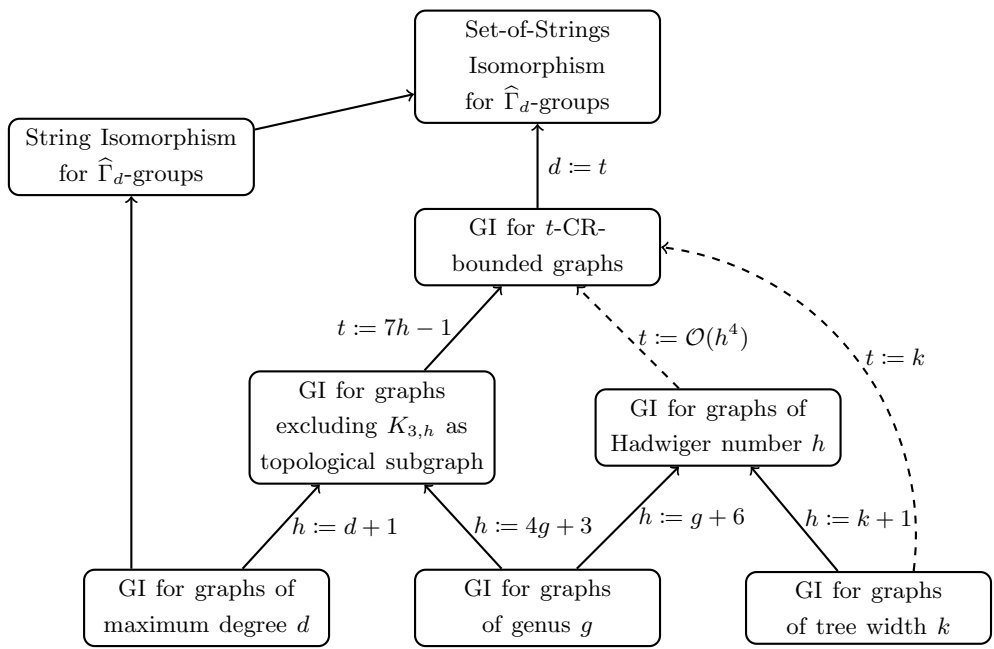

\subsection{Allowing Color Refinement to Split Small Classes}

As indicated above, the central tool for designing more efficient isomorphism tests based on the results of the previous section is the notion of $t$-CR-bounded graphs.
Intuitively speaking, a vertex-colored graph $(G,\chi)$ is \defword{$t$-CR-bounded} if the vertex-coloring can be transformed into a discrete coloring (i.e., a coloring where each vertex has its own color) by repeatedly applying the following operations:
\begin{enumerate}
 \item applying the Color Refinement algorithm, and
 \item picking a color class $[v]_\chi \coloneqq \{w \in V(G) \mid \chi(w) = \chi(v)\}$ for some vertex $v$ such that $|[v]_\chi| \leq t$ and assigning each vertex in the class its own color.
\end{enumerate}
The next definition formalizes this intuition.
For formal reason, it turns out to be more convenient to consider vertex- and arc-colored graphs.

\begin{Def}
 \label{def:t-cr-bounded}
 A vertex- and arc-colored graph $(G,\chi_V,\chi_E)$ is \defword{$t$-CR-bounded}
 if the sequence $(\chi_i)_{i \geq 0}$ reaches a discrete coloring where $\chi_0 \coloneqq \chi_V$,
 \[\chi_{2i+1} \coloneqq \ColRef{G,\chi_{2i},\chi_E}\]
 and
 \[\chi_{2i+2}(v) \coloneqq \begin{cases}
                     (v,1)              & \text{if } |[v]_{\chi_{2i+1}}| \leq t\\
                     (\chi_{2i+1}(v),0) & \text{otherwise}
                    \end{cases}\]
 for all $i \geq 0$.
 
 Also, for the minimal $i_\infty \geq 0$ such that $\chi_{i_\infty} \equiv \chi_{i_\infty+1}$, we refer to $\chi_{i_\infty}$ as the \defword{$t$-CR-stable} coloring of $G$ and denote it by $\tColRef{t}{G}$.
\end{Def}

In order to give an efficient isomorphism test for $t$-CR-bounded graphs we present a polynomial-time reduction to the Set-of-Strings Isomorphism Problem discussed in Section \ref{subsec:sets-of-strings}.

\begin{theorem}[\cite{Neuen20}]
 There is a polynomial-time Turing reduction from the Graph Isomorphism Problem for $t$-CR-bounded graphs to the Set-of-Strings Isomorphism Problem for $\mgamma_t$-groups.
\end{theorem}

\begin{proofsketch}
 Let $G$ be a $t$-CR-bounded graph and let $(\chi_i)_{i \geq 0}$ be the corresponding sequence of colorings.
 For simplicity we assume the arc-coloring of $G$ is trivial and restrict ourselves to computing a generating set for the automorphism group $\Aut(G)$.
 Let $\FP_i \coloneqq \pi(\chi_i)$ be the partition of the vertex set into the color classes of the coloring $\chi_i$.
 Observe that $\FP_i$ is defined in an isomorphism-invariant manner with respect to $G$ and $\FP_{i+1} \preceq \FP_i$ for all $i \geq 0$.
 
 The idea for the algorithm is to iteratively compute $\mgamma_t$-groups $\Gamma_i \leq \Sym(\FP_i)$ such that $(\Aut(G))[\FP_i] \leq \Gamma_i$ for all $i \geq 0$.
 First note that this gives an algorithm for computing the automorphism group of $G$.
 Since $G$ is $t$-CR-bounded there is some $i_\infty \geq 0$ such that $\FP_{i_\infty}$ is the discrete partition.
 Hence, $\Aut(G) \leq \Gamma_{i_\infty}$\footnote{
 Formally, this is not correct since $\Aut(G) \leq \Sym(V(G))$ and $\Gamma_{i_\infty} \leq \Sym(\{\{v\} \mid v \in V(G)\})$.
 However, the algorithm can simply identify $v$ with the singleton set $\{v\}$ to obtain the desired supergroup.}
 where $\Gamma_{i_\infty}$ is a $\mgamma_t$-group and the automorphism group of $G$ can be computed using Theorems \ref{thm:eq-sets-of-strings-hypergraphs} and \ref{thm:sets-of-strings}.
 
 For the sequence of groups the algorithm sets $\Gamma_0 \coloneqq \{\id\}$ to be the trivial group (this is correct since $\chi_0$ is the vertex-coloring of the input graph $G$).
 Hence, let $i > 0$ and suppose the algorithm already computed a generating set for $\Gamma_{i-1}$.
 
 If $i$ is even $\FP_i$ is obtained from $\FP_{i-1}$ by splitting all sets $P \in \FP_{i-1}$ of size at most $t$.
 Towards this end, the algorithm first updates $\Gamma_{i-1}' \coloneqq \Aut_{\Gamma_{i-1}}(\{\Fx_{i-1}\})$ where $\Fx_{i-1}(P) = |P|$ for all $P \in \FP_{i-1}$.
 Then, for each orbit $\FB$ of $\Gamma_{i-1}'$ there is some number $p \in \NN$ such that $|P| = p$ for all $P \in \FB$.
 Now $\Gamma_i$ is obtained from $\Gamma_{i-1}'$ by taking the wreath product with the symmetric group $S_p$ for all orbits where $p \leq t$.
 Clearly, $\Gamma_i$ is still a $\mgamma_t$-group.
 
 If $i$ is odd $\FP_i$ is obtained from $\FP_{i-1}$ by performing the
 Color Refinement algorithm.  Consider a single iteration of the Color
 Refinement algorithm.  More precisely, let $\FP$ be an invariant
 partition (with respect to $\Aut(G)$) and
 $\Gamma \leq \Sym(\FP)$ be a $\mgamma_t$-group such that
 $(\Aut(G))[\FP] \leq \Gamma$.  Let $\FP'$ be the partition obtained
 from $\FP$ by performing a single iteration of the Color Refinement
 algorithm.  We argue how to compute a $\mgamma_t$-group
 $\Gamma' \leq \Sym(\FP')$ such that $(\Aut(G))[\FP'] \leq \Gamma'$.
 Repeating this procedure for all iterations of the Color Refinement
 algorithm then gives the desired group $\Gamma_i$.

 In order to compute the group $\Gamma'$ consider the following
 collection of strings.  For every $v \in V(G)$ define
 $\Fx_v\colon\FP \rightarrow \NN^2$ via $\Fx_v(P) = (1,|N(v) \cap P|)$
 if $v \in P$ and $\Fx_v(P) = (0,|N(v) \cap P|)$ otherwise.  Observe
 that, by the Definition of the Color Refinement algorithm,
 $\Fx_v = \Fx_w$ if and only if there is some $P' \in \FP'$ such that
 $v,w \in P'$.  Hence, there is a natural one-to-one correspondence
 between $\FX \coloneqq \{\Fx_v \mid v \in V(G)\}$ and $\FP'$.  Now
 define $\Gamma'$ be the induced action of $\Aut_\Gamma(\FX)$ on the
 set $\FP'$ (obtained from this correspondence).  It is easy to verify
 that $(\Aut(G))[\FP'] \leq \Gamma'$.
 
 Clearly, all steps of the algorithm can be performed in polynomial
 time using an oracle to the Set-of-Strings Isomorphism Problem for
 $\mgamma_t$-groups.
\end{proofsketch}

We remark that the proof of the last theorem also implies that $\Aut(G) \in \mgamma_t$ for every $t$-CR-bounded graph $G$.
In combination with Theorem \ref{thm:sets-of-strings} we obtain an efficient isomorphism test for $t$-CR-bounded graphs.

\begin{cor}
 \label{cor:isomorphism-t-cr-bounded}
 The Graph Isomorphism Problem for $t$-CR-bounded graphs can be solved in time $n^{\polylog{t}}$.
\end{cor}

For the remainder of this section we shall exploit the algorithm from the corollary to design efficient isomorphism tests for a number of graph classes.
Towards this end, we typically build on another standard tool for isomorphism testing which is individualization of single vertices.
Intuitively, this allows us to break potential regularities in the input graph (for example, on a $d$-regular graph, the $t$-CR-stable coloring achieves no refinement unless $t \geq n$) and identify a ``starting point'' for analyzing the $t$-CR-stable coloring.

Let $G$ be a graph and let $X \subseteq V(G)$ be a set of vertices.
Let $\chi_V^*\colon V(G) \rightarrow C$ be the vertex-coloring
obtained from individualizing all vertices in the set $X$, i.e.,
$\chi_V^*(v) \coloneqq (v,1)$ for $v \in X$ and
$\chi_V^*(v) \coloneqq (0,0)$ for $v \in V(G) \setminus X$.  Let
$\chi \coloneqq \tColRef{t}{G,\chi_V^*}$ denote the $t$-CR-stable coloring
with respect to the input graph $(G,\chi_V^*)$.  We define the
\defword{$t$-closure} of the set $X$ (with respect to $G$) to be the set
\[\cl_t^G(X) \coloneqq \left\{v \in V(G) \mid |[v]_{\chi}| = 1\right\}.\]
Observe that $X \subseteq \cl_t^G(X)$.
For $v_1,\dots,v_\ell$ we use $\cl_t^G(v_1,\dots,v_\ell)$ as a shorthand for $\cl_t^G(\{v_1,\dots,v_\ell\})$

If $\cl_t^G(X) = V(G)$ then there is an isomorphism test for $G$ running in time $n^{|X| + \polylog{t}}$.
An algorithm first individualizes all vertices from $X$ creating $n^{|X|}$ many instances of GI for $t$-CR-bounded graphs each of which can be solved using Corollary \ref{cor:isomorphism-t-cr-bounded}.
This provides us a generic and powerful method for obtaining polynomial-time isomorphism tests for various classes of graphs.
As a first, simple example we argue that isomorphism for graphs of maximum degree $d$ can be tackled this way.

\begin{theorem}\label{theo:bounded-degree}
 Let $G$ be a connected graph of maximum degree $d$ and let $v \in V(G)$.
 Then $\cl_d^G(v) = V(G)$.
 
 In particular, there is a polynomial-time Turing reduction from the Graph Isomorphism Problem for graphs of maximum degree $d$ to the Graph Isomorphism Problem for $d$-CR-bounded graphs.
\end{theorem}

\begin{proof}
 Let $\chi \coloneqq \tColRef{d}{G,\chi_V^*}$ denote the $d$-CR-stable coloring where $\chi_V^*$ is the coloring obtained from individualizing $v$.
 For $i \geq 0$ let $V_i \coloneqq \{w \in V(G) \mid \dist_G(v,w) \leq i\}$.
 We prove by induction on $i \geq 0$ that $V_i \subseteq \cl_d^G(v)$.
 Since $V_n = V(G)$ this implies that $\cl_d^G(v) = V(G)$.
 
 The base case $i=0$ is trivial since $V_0 = \{v\} \subseteq \cl_d^G(v)$ as already observed above.
 So suppose $i \geq 0$ and let $w \in V_{i+1}$.
 Then there is some $u \in V_i$ such that $uw \in E(G)$.
 Moreover, $u \in \cl_d^G(v)$ by the induction hypothesis.
 This means that $|[u]_\chi| = 1$.
 Since $\chi$ is stable with respect to the Color Refinement algorithm $[w]_\chi \subseteq N(u)$.
 So $|[w]_\chi| \leq \deg(u) \leq d$.
 Hence, $|[w]_\chi| = 1$ because $\chi$ is $d$-CR-stable.
\end{proof}

This gives an isomorphism test for graphs of maximum degree $d$ running in time $n^{\polylog{d}}$.
Observe that the algorithm is obtained via a reduction to the Set-of-Strings Isomorphism Problem for $\mgamma_d$-groups.
We remark that, using similar, but slightly more involved ideas, there is also a polynomial-time Turing reduction from the isomorphism problem for graphs of maximum degree $d$ to the String Isomorphism Problem for $\mgamma_d$-groups \cite{Luks82,BabaiL83} (see also Section \ref{subsec:luks}).
Recall that the String Isomorphism Problem is a simpler problem which can be seen as a special case of the Set-of-Strings Isomorphism Problem (where the set of strings contains only one element).
Indeed, the original $n^{\polylog{d}}$ isomorphism test for graphs of maximum degree $d$ from \cite{GroheNS18} is obtained via this route.
However, considering the remaining applications of the isomorphism test for $t$-CR-bounded graphs in this section, such a behaviour seems to be an exception and reductions to the String Isomorphism Problem for $\mgamma_d$-groups are usually not known.
This highlights the significance of the Set-of-Strings Isomorphism Problem in comparison to the String Isomorphism Problem.

\subsection{Graphs of Small Genus}
\label{subsec:genus}

Next, we turn to the isomorphism problem for graphs of bounded Euler genus.
Recall that a graph has Euler genus at most $g$ if it can be embedded on a surface of Euler genus $g$.
We omit a formal definition of the genus of a graph and instead only rely on the following basic property.
A graph $H$ is a \defword{minor} of graph $G$ if $H$ can be obtained from $G$ by deleting vertices and edges as well as contracting edges.
A graph $G$ \defword{excludes $H$ as a minor} if it has no minor isomorphic to $H$.
It is well known that graphs of Euler genus $g$ exclude $K_{3,4g+3}$ (the complete bipartite graph with $3$ vertices on the left and $4g+3$ vertices on the right) as a minor \cite{Ringel65}.
The next lemma connects graphs that exclude $K_{3,h}$ as a minor to $t$-CR-bounded graphs.

Recall that a graph $G$ is \defword{3-connected} if $G - X$ is connected for every set $X \subseteq V(G)$ of size $|X| \leq 2$.

\begin{lemma}
 \label{la:extend-t-cr-bounded-scope}
 Let $(G,\chi)$ be a 3-connected, vertex-colored
 graph that excludes $K_{3,h}$ as a minor and suppose $V_1 \uplus V_2 = V(G)$ such that
 \begin{enumerate}
  \item $|[v]_\chi| = 1$ for all $v \in V_1$,
  \item $\chi$ is stable with respect to the Color Refinement algorithm, and
  \item $|V_1| \geq 3$.
 \end{enumerate}
 Then there exists $u \in V_2$ such that $|[u]_\chi| \leq h-1$.
\end{lemma}

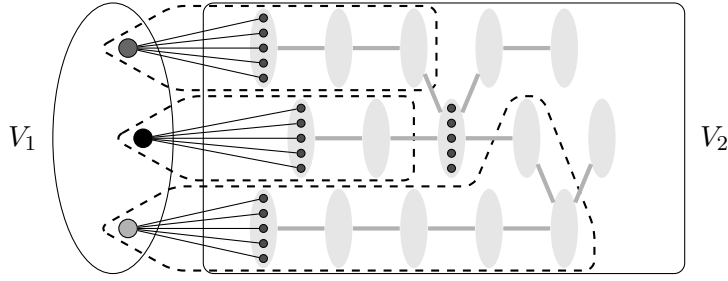
\begin{figure}
 \centering
 \begin{tikzpicture}
  \draw[thick,dashed,rounded corners] (0,0) -- (1.0,0.56) -- (4.0,0.56) -- (4.0,-0.56) -- (1.0,-0.56) -- cycle;
  \draw[thick,dashed,rounded corners] (-0.2,1.2) -- (0.8,1.76) -- (4.3,1.76) -- (4.3,0.64) -- (0.8,0.64) -- cycle;
  \draw[thick,dashed,rounded corners] (-0.2,-1.2) -- (0.8,-1.76) -- (6.4,-1.76) -- (6.4,-1.2) -- (5.7,0.56) -- (5.3,0.56) -- (4.8,-0.64) -- (0.8,-0.64) -- cycle;
  
  \draw (0,0) ellipse (0.8cm and 1.8cm);
  \draw[rounded corners] (1.2,-1.8) rectangle (7.6,1.8);
  
  \node at (-1.2,0) {$V_1$};
  \node at (8.0,0) {$V_2$};
  
  \node[ellipse,minimum height=30pt,minimum width=10pt,align=center,fill=gray!20] (v1) at (2.5,0) {};
  \node[ellipse,minimum height=30pt,minimum width=10pt,align=center,fill=gray!20] (v2) at (3.5,0) {};
  \node[ellipse,minimum height=30pt,minimum width=10pt,align=center,fill=gray!20] (v3) at (4.5,0) {};
  \node[ellipse,minimum height=30pt,minimum width=10pt,align=center,fill=gray!20] (v4) at (5.5,0) {};
  \node[ellipse,minimum height=30pt,minimum width=10pt,align=center,fill=gray!20] (v6) at (2,1.2) {};
  \node[ellipse,minimum height=30pt,minimum width=10pt,align=center,fill=gray!20] (v7) at (3,1.2) {};
  \node[ellipse,minimum height=30pt,minimum width=10pt,align=center,fill=gray!20] (v8) at (4,1.2) {};
  \node[ellipse,minimum height=30pt,minimum width=10pt,align=center,fill=gray!20] (v11) at (2,-1.2) {};
  \node[ellipse,minimum height=30pt,minimum width=10pt,align=center,fill=gray!20] (v12) at (3,-1.2) {};
  \node[ellipse,minimum height=30pt,minimum width=10pt,align=center,fill=gray!20] (v13) at (4,-1.2) {};
  \node[ellipse,minimum height=30pt,minimum width=10pt,align=center,fill=gray!20] (v14) at (5,-1.2) {};
  \node[ellipse,minimum height=30pt,minimum width=10pt,align=center,fill=gray!20] (v15) at (6,-1.2) {};
  
  \foreach \i/\j in {1/2,2/3,3/4,6/7,7/8,3/8,4/15,11/12,12/13,13/14,14/15}{
   \draw[thickedge] (v\i) edge (v\j);
  }
  
  \node[ellipse,minimum height=30pt,minimum width=10pt,align=center,fill=gray!20] (v5) at (6.5,0) {};
  \node[ellipse,minimum height=30pt,minimum width=10pt,align=center,fill=gray!20] (v9) at (5,1.2) {};
  \node[ellipse,minimum height=30pt,minimum width=10pt,align=center,fill=gray!20] (v10) at (6,1.2) {};
  
  \foreach \i/\j in {3/9,9/10,5/15}{
   \draw[thickedge] (v\i) edge (v\j);
  }
    
  \node[emptyvertex,fill=black!100] (w1) at (0.4,0) {};
  \node[emptyvertex,fill=black!60] (w2) at (0.2,1.2) {};
  \node[emptyvertex,fill=black!30] (w3) at (0.2,-1.2) {};
  
  \foreach \s in {-2,-1,0,1,2}{
   \node[tinyvertex,fill=black!70] (a\s) at ($(2.5,0)+(0,0.2*\s)$) {};
   \draw (w1) edge (a\s);
   \node[tinyvertex,fill=black!70] (b\s) at ($(2,1.2)+(0,0.2*\s)$) {};
   \draw (w2) edge (b\s);
   \node[tinyvertex,fill=black!70] (c\s) at ($(2,-1.2)+(0,0.2*\s)$) {};
   \draw (w3) edge (c\s);
  }
  
  \foreach \s in {-2,-1,0,1,2}{
   \node[tinyvertex,fill=black!70] (a\s) at ($(4.5,0)+(0,0.2*\s)$) {};
  }
 \end{tikzpicture}
 \caption{Visualization of the spanning tree $T$ described in the proof of Lemma \ref{la:extend-t-cr-bounded-scope}. Contracting the dashed regions into single vertices gives a subgraph isomorphic to $K_{3,h}$.}
 \label{fig:genus-to-cr-bounded}
\end{figure}

\begin{proof}
 Let $C \coloneqq \chi(V(G))$, $C_1 \coloneqq \chi(V_1)$ and $C_2 \coloneqq \chi(V_2)$.
 Also define $H$ to be the graph with vertex set $V(H) \coloneqq C$ and edge set
 \[E(H) = \{c_1c_2 \mid \exists v_1 \in \chi^{-1}(c_1), v_2 \in \chi^{-1}(c_2) \colon v_1v_2 \in E(G)\}.\]
 Let $C' \subseteq C_2$ be the vertex set of a connected component of $H[C_2]$.
 Then $|N_H(C')| \geq 3$ since each $v \in V_1$ forms a singleton color class with respect to $\chi$ and $G$ is $3$-connected.
 
 Now let $c_1,c_2,c_3 \in N_H(C')$ be distinct and also let $v_i \in \chi^{-1}(c_i)$ for $i \in [3]$.
 Also let $T$ be a spanning tree of $H[C' \cup \{c_1,c_2,c_3\}]$ such that $c_1,c_2,c_3 \in L(T)$ where $L(T)$ denotes the set of leaves of $T$ (see Figure \ref{fig:genus-to-cr-bounded}).
 Moreover, let $T'$ be the subtree of $T$ obtained from repeatedly removing all leaves $c \in C'$.
 Hence, $L(T') = \{c_1,c_2,c_3\}$.
 Then there is a unique color $c$ such that $\deg_{T'}(c) = 3$.
 Also, for $i \in [3]$, define $C_i'$ to be the set of internal vertices on the unique path from $c_i$ to $c$ in the tree $T'$.
 
 Since $|\chi^{-1}(c_i)| = 1$ and $\chi$ is stable with respect to the Color Refinement algorithm it holds that
 \[G\left[\chi^{-1}(C_i' \cup \{c_i\})\right]\]
 is connected.
 Let $U_i \coloneqq \chi^{-1}(C_i' \cup \{c_i\})$, $i \in [3]$.
 Also let $U = \chi^{-1}(c)$ and suppose that $|U| \geq h$.
 Then $N(U_i) \cap U \neq \emptyset$ by the definition of the tree $T$.
 Moreover, this implies $U \subseteq N(U_i)$ since $\chi$ is stable with respect to the Color Refinement algorithm.
 Hence, $G$ contains a minor isomorphic to $K_{3,h}$.
\end{proof}

\begin{cor}[\cite{Neuen20}]
 \label{cor:excluded-k3h-t-cr-bounded}
 Let $(G,\chi_V,\chi_E)$ be a $3$-connected, vertex- and arc-colored graph that excludes $K_{3,h}$ as a minor and let $v_1,v_2,v_3 \in V(G)$ be distinct vertices.
 Then $\cl_{h-1}^G(v_1,v_2,v_3) = V(G)$.
\end{cor}

\begin{proof}
 Let $\chi^*_V$ be the coloring obtained from $\chi_V$ after individualizing $v_1$, $v_2$ an $v_3$.
 Also let $\chi \coloneqq \tColRef{(h-1)}{G,\chi_V^*,\chi_E}$ be the $(h-1)$-CR-stable coloring with respect to the graph $(G,\chi_V^*,\chi_E)$.
 
 Suppose towards a contradiction that $\chi$ is not discrete (i.e., not every color class is a singleton).
 Let $V_2 \coloneqq \{v \in V(G) \mid |[v]_\chi| > 1\}$ and let $V_1 \coloneqq N_G(V_2)$.
 Then $|V_1| \geq 3$ since $|V(G) \setminus V_2| \geq 3$ and $G$ is $3$-connected.
 Also note that $\chi|_{V_1 \cup V_2}$ is a stable coloring for the graph $G[V_1 \cup V_2]$.
 Hence, by Lemma \ref{la:extend-t-cr-bounded-scope}, there is some vertex $u \in V_2$ such that $|[u]_\chi| \leq h-1$.
 Also $|[u]_\chi| > 1$ by the definition of $V_2$.
 But this contradicts the definition of an $(h-1)$-CR-stable coloring.
\end{proof}

Note that Corollary \ref{cor:excluded-k3h-t-cr-bounded} only provides an isomorphism test for
$3$-connected graphs excluding $K_{3,h}$ as a minor.  We can easily
remedy this problem by using the well-known fact that, for a
minor-closed class of graph, it suffices to solve the isomorphism
problem on vertex- and arc-colored $3$-connected graphs exploiting
the standard decomposition into $3$-connected components
\cite{hoptar72}.  This gives us an isomorphism test for graphs of
Euler genus at most $g$ running in time $n^{\polylog{g}}$.

\subsection{Graphs of Small Tree Width}

After discussing two fairly simple applications of the isomorphism test for $t$-CR-bounded graphs we now turn to the first more involved application considering graphs of small tree width.
In the following subsection we will extend some of the basic ideas discussed for graphs of small tree width further to deal with the more general case of excluding the complete graph $K_h$ as a minor.

We assume that the reader is familiar with the notion of tree width
(see, e.g., \cite[Chapter~12]{die16}, \cite[Chapter~7]{cygfomkow+15}).
For the sake of completeness, let us nevertheless give the basic definitions.
Let $G$ be a graph.
A \defword{tree decomposition} of $G$ is a pair $(T,\beta)$ where $T$ is a rooted tree and $\beta\colon V(T) \rightarrow 2^{V(G)}$, where $2^{V(G)}$ denotes the powerset of $V(G)$, such that
\begin{enumerate}
 \item[(T.1)] for every edge $vw \in E(G)$ there is some $t \in V(T)$ such that $v,w \in \beta(t)$, and
 \item[(T.2)] for every $v \in V(G)$ the set
   $\beta^{-1}(v)\coloneqq\{t \in V(T) \mid v \in \beta(t)\}$ is
   non-empty and connected in $T$, that is, the induced subgraph
   $T[\beta^{-1}(v)]$ is a subtree.
\end{enumerate}
The sets $\beta(t)$, $t \in V(T)$, are the \defword{bags} of the decomposition.
Also, the sets $\beta(t) \cap \beta(s)$, $st \in E(T)$, are the \defword{adhesion sets} of the decomposition.
The \defword{width} of a tree decomposition is defined as the maximum bag size minus one, i.e., $\width(T,\beta) \coloneqq \max_{t \in V(T)} |\beta(t)| - 1$.
The \defword{tree width} of a graph $G$, denoted $\tw(G)$, is the minimal width of any tree decomposition of $G$.

Our goal for this section is to design efficient isomorphism tests for
graphs of small tree width.  Once again, we want to build on the
isomorphism test for $t$-CR-bounded graphs.  However, compared to
graphs of small degree or genus, identifying parts of graphs of
bounded tree width that are $t$-CR-bounded (after fixing a constant
number of vertices) is more challenging.  This comes from the fact
that the automorphism group of a graph of tree width $k$ may contain
arbitrarily large symmetric groups.  For example, let $T_{d,h}$ denote
a complete $d$-ary tree of height $h$, where we should think of $d$ as
being unbounded and much larger than $k$. Also, let
$K_k \otimes T_{d,h}$ denote the graph obtained from $T_{d,h}$ by
replacing each vertex with a clique of size $k$ and two cliques being
completely connected whenever there is an edge in $T_{d,h}$.  Then
$\tw(K_k \otimes T_{d,h}) \leq 2k$, the graph $K_k \otimes T_{d,h}$
does not contain any $(k-1)$-separator, and the automorphism group
contains various symmetric groups $S_d$.

To circumvent this problem we exploit \defword{clique-separator
  decompositions} of \defword{improved graphs} which already form an
essential part of the first fixed-parameter tractable isomorphism test
for graphs of bounded tree width \cite{LokshtanovPPS17}, running in time $2^{\CO(k^5 \log k)}n^{5}$.
We first require some additional notation.
Let $G$ be a graph.
A pair $(A,B)$ where $A \cup B = V(G)$ is called a \defword{separation} if $E(A\setminus B,B\setminus A) = \emptyset$ (here, $E(X,Y) \coloneqq \{vw \in E(G) \mid v \in X,w \in Y\}$ denotes the set of edges with one endpoint in $X$ and one endpoint in $Y$).
In this case we call $A \cap B$ a \defword{separator}.
A separation $(A,B)$ is a called a \defword{clique separation} if $A \cap B$ is a clique and $A \setminus B \neq \emptyset$ and $B \setminus A \neq\emptyset$.
In this case we call $A \cap B$ a \defword{clique separator}.

\begin{Def}[\cite{Bodlaender03,LokshtanovPPS17}]
 The \defword{$k$-improvement} of a graph $G$ is the graph $G^k$ obtained from $G$ by adding an edge between every pair of non-adjacent vertices $v,w$
 for which there are more than $k$ pairwise internally vertex disjoint paths connecting~$v$ and~$w$.
 We say that a graph $G$ \defword{is $k$-improved} when $G^k=G$.
 
 A graph is \defword{$k$-basic} if it is $k$-improved and does not have any separating cliques.
\end{Def}

Note that a $k$-basic graph is $2$-connected.
We summarize several structural properties of~$G^k$.

\begin{lemma}[\cite{LokshtanovPPS17}]
 \label{la:k-improvement}
 Let $G$ be a graph and $k\in\NN$.
 \begin{enumerate}
  \item The $k$-improvement $G^k$ is $k$-improved, i.e., $(G^k)^k=G^k$.
  \item Every tree decomposition of $G$ of width at most $k$ is also a tree decomposition of $G^k$.
  \item There exists an algorithm that, given $G$ and $k$, runs in $\CO(k^2n^3)$ time and either correctly concludes that $\tw(G) > k$, or computes~$G^k$.
 \end{enumerate}
\end{lemma}

Since the construction of $G^k$ from $G$ is isomorphism-invariant, the
concept of the improved graph can be exploited for isomorphism
testing.  Given a graph, we can compute the $k$-improvement and assign
a special color to all added edges to distinguish them from the
original edges.  Hence, it suffices to solve the isomorphism problem
for $k$-improved graphs. In order to further reduce isomorphism testing to the case
of $k$-basic graphs, we build on a decomposition result of Leimer
\cite{Leimer93} which provides an isomorphism-invariant tree
decomposition of a graph into clique-separator free parts.

\begin{theorem}[\cite{Leimer93,ElberfeldS17}]
 \label{thm:clique-separator-decomposition}
 For every connected graph~$G$ there is an isomorphism-invariant tree
 decomposition $(T,\beta)$ of $G$, called clique separator
 decomposition, such that
 \begin{enumerate}
  \item for every $t \in V(T)$ the graph $G[\beta(t)]$ is clique-separator free, and
  \item each adhesion set of $(T,\beta)$ is a clique.
 \end{enumerate}
 Moreover, given a graph $G$, the clique separator decomposition of $G$ can be computed in polynomial time.
\end{theorem}

Observe that applying the theorem to a $k$-improved graph results in a
tree decomposition $(T,\beta)$ such that $G[\beta(t)]$ is $k$-basic
for every $t \in V(T)$.  Using standard dynamic programming arguments,
this decomposition allows us to essentially restrict to the case of
$k$-basic graphs at the cost of an additional factor
$2^{\CO(k \log k)}n^{\CO(1)}$ in the running time.  Now, a key insight
is that $k$-basic graphs are once again $k$-CR-bounded after
individualizing a single vertex.  The next lemma provides us the main
tool for proving this.  Recall that, for a graph $G$ and sets
$A,B \subseteq V(G)$, we denote $E(A,B) \coloneqq \{vw \in E(G) \mid v \in A, w \in B\}$.

\begin{lemma}
 \label{la:many-disjoint-paths}
 Let $G$ be a graph and let $\chi$ be a coloring that is stable with respect to the Color Refinement algorithm.
 Let $X_1,\dots,X_m$ be distinct color classes of $\chi$ and suppose that $E(X_i,X_{i+1}) \neq \emptyset$ for all $i \in [m-1]$.
 Then there are $\min_{i \in [m]}|X_i|$ many vertex-disjoint paths from $X_1$ to $X_m$.
\end{lemma}

To gain some intuition on the lemma, consider the simple case where
$\ell \coloneqq |X_i| = |X_j|$ for all $i,j \in [m]$.  Since $\chi$ is
stable we conclude that the bipartite graph
$H_i \coloneqq (X_i \cup X_{i+1},E(X_i,X_{i+1}))$ is non-empty and
biregular.  Hence, by Hall's Marriage Theorem, there is a perfect
matching $M_i$ in $H_i$ of size $|M_i| = \ell$.  Taking the disjoint
union of all sets $M_i$, $i \in [m-1]$, gives $\ell$ vertex-disjoint
paths from $X_1$ to $X_m$.  The general case can be proved using
similar arguments.  The details can be found in \cite{GroheNW20}.

\begin{lemma}
 \label{la:k-basic-CR-bounded}
 Let $G$ be a $k$-basic graph.
 Also let $v \in V(G)$ such that $\deg(v) \leq k$.
 Then $\cl_k^G(v) = V(G)$.
\end{lemma}

\begin{proof}
 Let $\chi$ denote the $k$-CR-stable coloring of $G$ after individualizing $v$.
 First note that $N(v) \subseteq \cl_k^G(v)$ since $\deg(v) \leq k$.
 Now suppose towards a contradiction that $\cl_k^G(v) \neq V(G)$ and let $W$ be the vertex set of a connected component of $G - \cl_k^G(v)$.
 Then $v \notin N(W)$.
 Since $G$ does not have any separating cliques there exist $w_1,w_2 \in N(W)$ such that $w_1w_2 \notin E(G)$.
 Let $u_0,u_1,\dots,u_m,u_{m+1}$ be a path from $w_1 = u_0$ to $w_2 = u_{m+1}$ such that $u_i \in W$ for all $i \in [m]$.
 Let $X_i \coloneqq [u_i]_\chi$ and note that $|X_i| \geq k+1$.
 Also note that $X_1 \subseteq N(w_1)$ and $X_m \subseteq N(w_2)$ since $w_1,w_2 \in \cl_k^G(v)$ and $\chi$ is stable.
 Hence, by Lemma \ref{la:many-disjoint-paths}, there are $k+1$ internally vertex-disjoint paths from $w_1$ to $w_2$.
 But this contradicts the fact that $G$ is $k$-improved.
\end{proof}

Overall, this gives an isomorphism test for graphs of tree width $k$
running in time $2^{\CO(k \log k)} n^{\polylog{k}}$ as follows.
First, the algorithm computes the clique-separator decomposition of
the $k$-improved graph using Lemma \ref{la:k-improvement} and Theorem
\ref{thm:clique-separator-decomposition}.  The isomorphism problem for
the $k$-basic parts can be solved in time $n^{\polylog{k}}$ by Lemma
\ref{la:k-basic-CR-bounded} and Corollary
\ref{cor:isomorphism-t-cr-bounded}.  Finally, the results for the bags
can be joined using dynamic programming along the clique separator
decomposition.  (Formally, this requires the algorithm to solve a more
general problem for $k$-basic graphs where information obtained by
dynamic programming on the adhesion sets to the children of a bag are
incorporated.  But this can be easily achieved within the
group-theoretic framework; see e.g.\ \cite{GroheNW20}.)

Using further advanced techniques (which are beyond the scope of this
article) this result can be strengthened into two directions.  First,
building isomorphism-invariant decompositions of $k$-basic graphs, it
is possible to build an improved fixed-parameter tractable isomorphism test for graphs of tree width
$k$ running in time $2^{k \polylog{k}}n^3$ \cite{GroheNSW20}
(improving on the $2^{\CO(k^5\log k)}n^5$ algorithm by Lokshtanov et al.\ \cite{LokshtanovPPS17}).

And second, isomorphism of graphs of tree width $k$ can also be tested in time $n^{\polylog{k}}$.
Here, the crucial step is to speed up to dynamic programming procedure
along the clique separator decomposition.  Let $D$ be the vertex set
of the root bag of the clique separator decomposition and let
$Z_1,\dots,Z_\ell$ be connected components of $G - D$ such that
$S \coloneqq N_G(Z_i) = N_G(Z_j)$ for all $i,j \in [\ell]$.  Observe
that $|S| \leq k+1$ since it is a clique by the properties of the
clique separator decomposition.  After recursively computing
representations for the sets of isomorphisms between $(G[Z_i \cup S],S)$
and $(G[Z_j \cup S],S)$ for all $i,j \in [\ell]$ the partial solutions
have to be joined into the automorphism group of
$(G[Z_1 \cup \dots \cup Z_\ell \cup S],S)$.  A simple strategy is to
iterate through all permutations $\gamma \in \Sym(S)$ and check in
polynomial time which permutations can be extended to an automorphism.
This gives a running time $2^{\CO(k \log k)} \ell^{\CO(1)}$ as already
exploited above.  Giving another extension of Babai's quasi-polynomial
time algorithm, Wiebking \cite{Wiebking20} presented an algorithm
solving this problem in time $\ell^{\polylog{k}}$.  This enables us to
achieve the running time $n^{\polylog{k}}$ for isomorphism testing of
graphs of tree width $k$.

\subsection{Graphs Excluding Small (Topological) Minors}

In the previous two subsections we have seen isomorphism tests running
in time $n^{\polylog{g}}$ for graphs of genus $g$ and
$n^{\polylog{k}}$ for graphs of tree width $k$.  Both algorithms
follow a similar pattern of computing an isomorphism-invariant
tree-decomposition of small adhesion-width such that (the torso of)
each bag is $t$-CR-bounded after individualizing a constant number of
vertices.  Recall that the same strategy also works for graphs of
bounded degree (taking the trivial tree decomposition with a single
bag).  This suggests that a similar approach might work in a much more
general setting considering graphs that exclude an arbitrary
$h$-vertex graph as a minor or even as a topological subgraph.

Let $G,H$ be graphs.  Recall that $H$ is a minor of $G$ if $H$ can be
obtained from $G$ by deleting vertices and edges as well as
contracting edges.  Moreover, $H$ is a \defword{topological subgraph} of $G$
if $H$ can be obtained from $G$ by deleting vertices and edges as well
as dissolving degree-two vertices (i.e., deleting the vertex and
connecting its two neighbors).  Note that, if $H$ is a topological subgraph of $G$,
then $H$ is also a minor of $G$ (in general, this is not
true in the other direction).  A graph $G$ \defword{excludes $H$ as a
(topological) minor} if it has no (topological) minor isomorphic to $H$.

A common argument for the correctness of the algorithms considered so
far is the construction of forbidden subgraphs from $t$-CR-stable
colorings that are not discrete.  In order to generalize the existing
algorithms we first introduce a tool that allows us to construct large
(topological) minors from $t$-CR-stable colorings which are not
discrete, but still contain some individualized vertices.  Let $G$ be
a graph, let $\chi\colon V(G) \rightarrow C$ be a vertex-coloring and
let $T$ be a tree with vertex set $V(T)=C$.  A subgraph
$H \subseteq G$ \defword{agrees} with $T$ if
$\chi|_{V(H)} \colon H \cong T$, i.e., the coloring $\chi$ induces an
isomorphism between $H$ and $T$.  Observe that each $H\subseteq G$
that agrees with a tree $T$ is also a tree.

For a tree $T$ we define $V_{\leq i}(T) \coloneqq \{t \in V(T) \mid \deg(t) \leq i\}$ and $V_{\geq i}(T) \coloneqq \{t \in V(T) \mid \deg(t) \geq i\}$.
It is well known that for trees $T$ it holds that $|V_{\geq 3}(T)|<|V_{\leq 1}(T)|$.

\begin{lemma}[\cite{GroheNW20}]
 \label{la:find-disjoint-trees}
 Let $G$ be a graph and let $\chi\colon V(G)\to C$ be a vertex-coloring that is stable with respect to the Color Refinement algorithm.
 Also let $T$ be a tree with vertex set $V(T) = C$ and assume that for every $c_1c_2\in E(T)$ there are $v_1 \in \chi^{-1}(c_1)$ and $v_2 \in \chi^{-1}(c_2)$ such that $v_1v_2 \in E(G)$.
 Let $m \coloneqq \min_{c\in C} |\chi^{-1}(c)|$ and let $\ell \coloneqq 2|V_{\leq 1}(T)|+|V_{\geq 3}(T)|$.
 
 Then there are (at least) $\lfloor\frac{m}{\ell}\rfloor$ pairwise vertex-disjoint trees in $G$ that agree with $T$.
\end{lemma}

Intuitively speaking, assuming there is a stable coloring in which all color classes are large, the lemma states that one can find a large number of vertex-disjoint trees with predefined color patterns.
Observe that Lemma \ref{la:find-disjoint-trees} generalizes Lemma \ref{la:many-disjoint-paths} which provides such a statement for paths instead of arbitrary trees (ignoring the additional factor of two for the number of vertex-disjoint paths in comparison to the minimum color class size).

The proof of this lemma is quite technical and we omit any details here.
As a first, simple application we can provide a strengthened version of Lemma \ref{la:extend-t-cr-bounded-scope}.

\begin{lemma}
 \label{la:extend-t-cr-bounded-scope-topological}
 Let $(G,\chi)$ be a 3-connected, vertex-colored graph that excludes $K_{3,h}$ as a topological subgraph and suppose $V_1 \uplus V_2 = V(G)$ such that
 \begin{enumerate}
  \item $|[v]_\chi| = 1$ for all $v \in V_1$,
  \item $\chi$ is stable with respect to the Color Refinement algorithm, and
  \item $|V_1| \geq 3$.
 \end{enumerate}
 Then there exists $u \in V_2$ such that $|[u]_\chi| \leq 7h-1$.
\end{lemma}

In comparison with Lemma \ref{la:extend-t-cr-bounded-scope} the graph $G$ is only required to exclude $K_{3,h}$ as a topological subgraph instead of a minor.
Hence, Lemma \ref{la:extend-t-cr-bounded-scope-topological} covers a much larger class of graphs.
Observe that the bound on the size of the color class $[u]_\chi$ is slightly worse which however is not relevant for most algorithmic applications since the additional factor disappears in the $\CO$-notation.

The proof of Lemma \ref{la:extend-t-cr-bounded-scope-topological} is similar to the proof of Lemma \ref{la:extend-t-cr-bounded-scope}.
Recall that in the proof of Lemma~\ref{la:extend-t-cr-bounded-scope} we constructed a tree $T'$ with three leaves representing three vertices $v_1,v_2,v_3 \in V_1$.  
Instead of using $T'$ to construct a minor $K_{3,h}$ this tree can also be used to construct a topological subgraph $K_{3,h}$ building on Lemma \ref{la:find-disjoint-trees}.
Note that $2|V_{\leq 1}(T')|+|V_{\geq 3}(T')| \leq 7$ which gives the additional factor on the bound of $|[u]_\chi|$.

Following the same arguments as in Section \ref{subsec:genus} we obtain a polynomial-time Turing reduction from GI for graphs excluding $K_{3,h}$ as a topological subgraph to GI for $(7h-1)$-CR-bounded graphs.

\begin{theorem}
 The Graph Isomorphism Problem for graphs excluding $K_{3,h}$ as a topological subgraph can be solved in time $n^{\polylog{h}}$.
\end{theorem}

This generalizes the previous algorithms for graphs of bounded degree as well as graphs of bounded genus (see Figure \ref{fig:reductions}).

Next, we wish to generalize the algorithms for graphs of bounded genus and graphs of bounded tree width.
More precisely, we aim for an algorithm that decides isomorphism of graphs that exclude an arbitrary $h$-vertex graph as a minor in time $n^{\polylog{h}}$.
Towards this end, first observe that a graph $G$ excludes an arbitrary $h$-vertex graph as a minor if and only if $G$ excludes $K_h$ as a minor.
The maximum $h$ such that $K_h$ is a minor of $G$ is also known as the \defword{Hadwiger number} of $G$.
Note that a graph of Hadwiger number $h$ excludes $K_{h+1}$ as a minor.
So in other words, the goal is to give an isomorphism test for graphs of Hadwiger number $h$ running in time $n^{\polylog{h}}$.

As a first basic tool we argue that, given a set $X \subseteq V(G)$, the $t$-closure of $X$ only stops to grow at a separator of small size.
Note that similar properties are implicitly proved for graphs of small genus and graphs of small tree width in the previous two subsections.

\begin{lemma}[\cite{GroheNW20}]
 \label{la:closure-small-separator}
 Let $t \geq 3h^3$.
 Let $G$ be a graph that excludes $K_h$ as a topological subgraph.
 Also let $X \subseteq V(G)$ and let $Z$ be the vertex set of a connected component of $G - \cl_t^G(X)$.
 Then $|N(Z)| < h$.
\end{lemma}

The strategy for the proof is very similar to the previous arguments building on Lemma \ref{la:find-disjoint-trees}.
We assume towards a contradiction that $|N(Z)| \geq h$ for some connected component $Z$ of $G - \cl_t^G(X)$ and construct a topological subgraph $K_h$.
Towards this end, it suffices to find vertex-disjoint paths between all pairs of vertices of a subset $Y \subseteq N(Z)$ of size at least $h$.
Using Lemma \ref{la:find-disjoint-trees} we can actually construct a large number of vertex-disjoint trees such that each $v \in Y$ has a neighbor in all of those trees.
In particular, this gives the desired set of vertex-disjoint paths.

Now, similar to the previous examples, the central step is to obtain a suitable isomorphism-invariant decomposition of the input graph so that each part of the decomposition is $t$-CR-bounded after individualizing a constant number of vertices.
For graphs of bounded genus as well as graphs of bounded tree width we could rely on existing decompositions from the literature.
However, for graphs excluding $K_h$ as a minor, such results do not seem to be known.
In particular, neither of the decompositions applied for graphs of bounded genus and graphs of bounded tree width are suitable for graph classes that only exclude $K_h$ as a minor.

The central idea to resolve this problem is to use the $t$-closure of a set itself as a way of defining a graph decomposition.
We have already established in the last lemma that the $t$-closure can only stop to grow at a separator of small size.
In particular, this enables us to bound the adhesion-width of the resulting tree decomposition (i.e., the maximum size of an adhesion set).
The second, more involved challenge is to obtain a decomposition which is isomorphism-invariant.
Recall that, in all previous examples, we individualized a constant number of vertices to ``initiate the growth of the $t$-closure''.
In order to still obtain an isomorphism-invariant decomposition we have to ensure that $\cl_t^G(v)$ is independent of the specific choice of $v$.
Towards this end, we want to argue that there is an isomorphism-invariant set $X \subseteq V(G)$ such that $\cl_t^G(v) = \cl_t^G(w)$ for all $v,w \in X$.

This basic idea can indeed be realized, but the concrete formulation turns out to be slightly more complicated.
More precisely, in order to incorporate information not ``visible'' to the Color Refinement algorithm, the $t$-closure is computed on an extended graph which allows us to explicitly encode certain structural information of the input graph.

\begin{lemma}[\cite{GroheNW20}]
 \label{la:initial-color}
 Let $t = \Omega(h^4)$.  There is a polynomial-time algorithm that,
 given a connected vertex-colored graph $G$, either correctly
 concludes that $G$ has a minor isomorphic to $K_h$ or computes a
 vertex-colored \footnote{In \cite[Theorem~IV.2]{GroheNW20}, this graph is
   pair-colored, i.e., a color is assigned to each pair of
   vertices. By adding gadgets for the pair colors such a graph can
   easily be transformed into a vertex-colored graph while preserving
   all required properties.}  graph $(G',\chi')$ and a set
 $\emptyset \neq X \subseteq V(G')$ such that
 \begin{enumerate}
  \item\label{item:initial-color-1} $X = \{v \in V(G') \mid \chi'(v) = c\}$ for some color $c$,
  \item\label{item:initial-color-2} $X \subseteq \cl_t^{(G',\chi')}(v)$ for every $v \in X$, and
  \item\label{item:initial-color-3} $X \subseteq V(G)$.
 \end{enumerate}
 Moreover, the output of the algorithm is isomorphism-invariant with respect to $G$.
\end{lemma}

The proof of the lemma is rather technical and builds on the $2$-dimensional Weisfeiler-Leman algorithm.
For details we refer the reader to \cite{GroheNW20}.

Now, an algorithm for testing isomorphism of graphs excluding $K_h$ as a minor can proceed as follows.
Given a graph $G$, we first compute an isomorphism-invariant set $X \subseteq V(G)$ using Lemma \ref{la:initial-color} as well as its closure $D \coloneqq \cl_t^G(X)$.
Then $|N_G(Z)| < h$ for every connected component $Z$ of the graph $G - D$ by Lemma \ref{la:closure-small-separator}.
Actually, by proceeding recursively on all connected components of $G - D$ we obtain an isomorphism-invariant tree decomposition $(T,\beta)$ of adhesion-width at most $h-1$.
Moreover, for each bag, we can consider the corresponding extended graph from Lemma \ref{la:initial-color} to obtain a structure associated with the bag that is $t$-CR-bounded for $t = \Omega(h^4)$ (after individualizing a single vertex).
Hence, using similar techniques as for graphs of bounded tree width in the previous subsection, we obtain the following result.

\begin{theorem}[Grohe, Neuen and Wiebking~\cite{GroheNW20}]
 The Graph Isomorphism Problem for graphs excluding $K_h$ as a minor can be solved in time $n^{\polylog{h}}$.
\end{theorem}

From the algorithm described above we can also infer some information on the structure of automorphism groups of graphs excluding $K_h$ as a minor.

\begin{theorem}
 \label{thm:structure-aut}
 Let $G$ be a graph that excludes $K_h$ as a minor.
 Then there is an isomorphism-invariant tree decomposition $(T,\beta)$ of $G$ such that
 \begin{enumerate}
  \item the adhesion-width of $(T,\beta)$ is at most $h-1$, and
  \item for all $t \in V(T)$ there exists $v \in \beta(t)$ such that $(\Aut(G))_v[\beta(t)] \in \mgamma_d$ for $d = \CO(h^4)$.
 \end{enumerate}
\end{theorem}

\section{Concluding Remarks}

We survey recent progress on the Graph Isomorphism Problem. Our focus is on
Babai's quasi-polynomial time algorithm and a series of subsequent
quasi-polynomial parameterized algorithms. In particular, we
highlight the power of a new technique centered around the notion of
$t$-CR-boundedness. Actually, upon completion of this survey article, the second author
showed that these techniques can be pushed even further obtaining an isomorphism test running in
time $n^{\polylog{h}}$ for graphs only excluding a graph of order $h$ as a topological subgraph \cite{Neuen20a}.
More precisely, building on a very lengthy and technical analysis, it is possible to generalize
Lemma \ref{la:initial-color} to graphs only excluding $K_h$ as a topological subgraph.
Following the same strategy outlined above, this implies a faster isomorphism test graphs excluding $K_h$ as a topological subgraph.

Research on the isomorphism problem is still very active, and
interesting questions remain open. Of course, the most important
question remains whether isomorphism of graphs can be tested in polynomial time.
Let us also mention a few more specific questions that seem to be within reach.
\begin{enumerate}
 \item[1.] Can isomorphism of hypergraphs be tested in time $n^{\polylog{n}}m^{\CO(1)}$
  where $n$ denotes the number of vertices and $m$ the number of hyperedges of the input?
 \item[2.] Can the techniques based on $t$-CR-boundedness be exploited to design faster isomorphism tests for dense graphs?
  For example, is there an isomorphism algorithm running in time $n^{\polylog{k}}$
  for graphs of rank width at most $k$?
 \item[3.] Is GI parameterized by the Hadwiger number fixed-parameter
  tractable, that is, is there an isomorphism algorithm running in
  time $f(k)n^{O(1)}$ for graphs excluding a graph of order $k$ as a
  minor, for some arbitrary function $f$? The same question can be
  asked for topological subgraphs.
 \item[4.] Is GI parameterized by maximum degree number fixed-parameter
  tractable?
\end{enumerate}
The fourth question and the third for topological subgraphs seem to
be hardest. The question whether GI parameterized by maximum degree is
fixed parameter tractable has already been raised more than twenty
years ago in Downey and Fellow's monograph on parameterized complexity
\cite{dowfel99}.

Other interesting research directions include the Group Isomorphism Problem (see, e.g., \cite{BabaiCQ12,BrachterS20,BrooksbankGLQW19,Rosenbaum20}) and the Tournament Isomorphism Problem (see, e.g., \cite{BabaiL83,Schweitzer17}).

\bibliographystyle{amsplain}
\bibliography{literature}

\myaddress

\end{document}